\documentclass[a4paper,11pt]{amsart}
\usepackage{amssymb}
\usepackage{amsmath}
\usepackage{mathtools}
\usepackage[utf8]{inputenc}
\usepackage{amsthm}
\usepackage{amsfonts}
\usepackage{enumerate}
\usepackage{mathrsfs}
\usepackage{enumitem}
\usepackage{arydshln}
\usepackage[ruled,vlined]{algorithm2e}
\usepackage{hyperref}
\usepackage{pdfpages}

\usepackage{pgfplots}
\pgfplotsset{compat=1.15}
\usetikzlibrary{arrows}
\usepackage{caption}
\usepackage[labelformat=simple]{subcaption}

\usepackage{xcolor}

\newlength\mylen

\usepackage{young}
\usepackage[vcentermath]{youngtab}

\hypersetup{
	colorlinks   = true, 
	urlcolor     = blue, 
	linkcolor    = blue, 
	citecolor   = red 
}

\usepackage[a4paper]{geometry}
\usepackage{titletoc}
\usepackage{fancyhdr}
\newtheorem{thm}{Theorem}
\newtheorem{prop}[thm]{Proposition}
\newtheorem{lem}[thm]{Lemma}

\newtheorem*{question*}{Question}
\theoremstyle{definition}
\newtheorem{defi}[thm]{Definition}

\theoremstyle{remark} 
\newtheorem{rem}[thm]{Remark}
\newtheorem*{rem*}{Remark}
\theoremstyle{definition}

\newtheorem*{nota*}{Notation}

\newcommand{\X}{{\mathfrak X}}
\newcommand{\GL}{\mathrm{GL}}
\newcommand{\K}{{\mathbf K}}
\newcommand{\KK}{{\overline{\K}}}

\newcommand{\Q}{\overline{\mathbb Q}}
\newcommand{\QQ}{\mathbb Q}
\newcommand{\C}{\mathbb{C}}
\newcommand{\Z}{\mathbb{Z}}
\newcommand{\N}{\mathbb{N}}
\renewcommand{\u}{{\boldsymbol u}}
\renewcommand{\v}{{\boldsymbol v}}

\newcommand{\lambd}{{\boldsymbol{\lambda}}}

\renewcommand{\a}{{\boldsymbol{a}}}
\newcommand{\bxi}{\boldsymbol{\xi}}
\newcommand{\omeg}{{\boldsymbol{\omega}}}
\newcommand{\w}{{\boldsymbol w}}
\newcommand{\e}{{\boldsymbol e}}

\newcommand{\f}{{\boldsymbol f}}

\newcommand{\x}{{\boldsymbol x}}

\newcommand{\puip}{\phi_p}

\newcommand{\I}{\mathcal{I}}

\newcommand{\U}{\mathfrak{U}}

\newcommand{\val}{\operatorname{val}}

\newcommand{\vb}{\val A^{-1}}

\newcommand{\Hahn}{\mathscr{H}}

\newcommand{\Puis}{\mathscr{P}}
\newcommand{\logm}{{\ell}}

\begin{document}
	\title[]{Computing basis of solutions of any Mahler equation}

	\author{Colin Faverjon}
	\address{LAMFA, UMR 7352, Universit\'e Picardie Jules Vernes, 33, rue saint-Leu, 80039 Amiens, France}
	\email{colin.faverjon@math.cnrs.fr}
	\author{Marina Poulet}
	\address{Institut Fourier, UMR 5582, Laboratoire de Math\'ematiques, Universit\'e Grenoble Alpes, CS 40700, 38058 Grenoble cedex 9, France}
	\email{marina.poulet@univ-grenoble-alpes.fr}
	
	\date{\today}

	\begin{abstract} 
		Mahler equations arise in a wide range of contexts including the study of finite automata, regular sequences, algebraic series over $\mathbb F_p(z)$, and periods of Drinfeld modules. Introduced a century ago by K. Mahler to study the transcendence of certain complex numbers, they have recently been the subject of several works establishing a deep connection between such transcendence properties and the nature of their solutions. While numerous studies have investigated these solutions, existing algorithms can only compute them in specific rings: rational functions, power series, Puiseux series, or Hahn series. This paper solves the problem by providing an algorithm that computes a complete basis of solutions for any Mahler equation, along with a decomposition of each solution over the field of Puiseux series. Along the way, we describe an algorithm that computes a fundamental matrix of solutions for any Mahler system.
	\end{abstract}
	\maketitle
	
	\vspace{-1.5cm}
	\setcounter{tocdepth}{1}
	\setcounter{secnumdepth}{3}
	\titlecontents{section}
	[0.5em]
	{}
	{\contentslabel{1.3em}}
	{\hspace*{-2.3em}}
	{\titlerule*[1pc]{.}\contentspage}
	{\addvspace{2em}\bfseries\large}
	\titlecontents{subsection}
	[3.8em]
	{} 
	{\contentslabel{2em}}
	{\hspace*{-3.2em}}
	{\titlerule*[1pc]{.}\contentspage}
	\titlecontents{subsubsection}
	[6.1em]
	{} 
	{\contentslabel{2.4em}}
	{\hspace*{-4.1em}}
	{\titlerule*[1pc]{.}\contentspage}
	\renewcommand{\contentsname}{Contents}
	
	\pdfbookmark[0]{\contentsname}{Contents}
	\tableofcontents

	\section{Introduction}
	
	The study of solutions of linear differential equations is a classical and extensive field of research. Any linear differential equation of order $m$ with coefficients in $\mathbb{C}((z))$ admits a basis of solutions consisting of $\mathbb{C}$-linear combinations of elements of the form
	\begin{equation}\label{eq:diff}
	f(z) z^\alpha (\log z)^k e^{P(z^{-1/m!})}
	\end{equation}
	where $f(z)$ is a Gevrey series, $\alpha \in \mathbb{C}$, $k \in \mathbb{Z}_{\geq 0}$, and $P \in X\C[X]$. Various algorithms have been established for computing such solutions \cite{Tur55, Tou87, Ba97}. 
	
	The subject of our paper is the so-called Mahler equations.
	A linear Mahler equation is an equation of the form
	$$
	a_0(z)y(z)+a_1(z)y(z^p)+\cdots+a_m(z)y(z^{p^m})=0
	$$
	where $p \geq 2$ is an integer, $a_0a_m\neq0$ and $a_0,\ldots,a_m$ are, depending on the context, polynomials, power series, etc. Though less widely known than their differential counterparts, they have been the subject of numerous recent papers due to their connections with several areas, see the discussion in Section~\ref{sec:motivations}. Only quite recently, a general form for a basis of solutions analogous to \eqref{eq:diff} has been established for Mahler equations (see \cite{FR25}). However, the proof is nonconstructive. Algorithms do exist to compute rational, power series, Puiseux series \cite{CDDM18}, or Hahn series \cite{FR24} solutions of a given Mahler equation, but none can provide a full basis of solutions. Our paper fills this gap.

	\subsection{Linear Mahler equations}
	Let $\K$ be a field and $z$ be an indeterminate. For each integer $p\geq 2$, we consider the map $\puip : \K(z) \to \K(z)$ defined by $z \mapsto z^p$. Let $\mathcal R$ be a ring extension of $\K(z)$ and let $\puip$ extend to $\mathcal R$. A $p$-Mahler equation over $\mathcal R$ is a linear functional equation of the form
	\begin{equation}\label{eq:Mahler}
	a_0y+a_1\puip(y)+\cdots + a_m\puip^m(y) =0
	\end{equation}
	where $a_0,\ldots,a_m \in \mathcal R$ and $a_0a_m \neq 0$. In what follows, we focus on the case where $\mathcal R = \K[[z]]$. However, the methods we develop here extend immediately to the case where $\mathcal R$ is the field of Puiseux series over $\K$. 
	
	Power series --or even Puiseux series-- are not sufficient to solve arbitrary Mahler equations over $\K[[z]]$. A famous example comes from the equation
	\begin{equation*}
	y(z^p)-y(z)=z^{-1}
	\end{equation*}
	for which a quick computation reveals that it has no Puiseux series solution. Meanwhile, the formal series $h(z)=\sum_{k\geq 1} z^{-1/p^k}$ is a solution to this equation. This is not a Puiseux series, but since its support is a well-ordered subset of the rational numbers, $h(z)$ is a \textit{Hahn series}. Generalizing this remark, in \cite{FR25}, the authors proved that\footnote{While the results in \cite{FR25} are stated for $\KK:=\Q$, those invoked here hold over the algebraic closure of any field. Since we must reprove these results for algorithmic considerations, we provide details later in the paper.}, given an integer $s\geq 1$, a $s$-dimensional multi-linear recurrence sequence $\u=(u_{k_1,\ldots,k_s})_{k_1,\ldots,k_s}\in \KK^{\Z^{s}}$ and a $s$-tuple of positive rational numbers $\a=(a_1,\ldots,a_s)$, the Hahn series
	\begin{equation}
	\boldsymbol{\xi}_{\omeg} := \sum_{k_1,\cdots,k_s\geq 1} u_{k_1,\ldots,k_s}z^{-\frac{a_1}{p^{k_1}}-\cdots-\frac{a_s}{p^{k_1+\cdots+k_s}}}\quad \text{with }\boldsymbol{\omega}:=(\u,\a)
	\end{equation}
	is a solution to some Mahler equation over $\KK[[z]]$, where we let $\KK$ denote an algebraic closure of $\K$. We let $\Omega$ denote the set of all such tuples $\boldsymbol{\omega}$. For convenience, when $s=0$ we let $\bxi_\omeg=1$. The definition and some details about linear recurrent sequences are given in Section~\ref{sec:Hahn}. When $\KK$ has characteristic $0$, any linear recurrent sequence can be uniquely written as a linear combination of sequences with general term 
	\begin{equation}\label{eq:basis_rec_lin}
	k_1^{\alpha_1}\cdots k_s^{\alpha_s}\lambda_1^{k_1}\cdots \lambda_s^{k_s}, \quad \alpha_1,\ldots,\alpha_s \in \Z_{\geq 0},\,\lambda_1,\ldots,\lambda_s \in \KK^\times.
	\end{equation}
	Hence, we recover the Hahn series introduced in \cite{FR25}. Such a decomposition no longer holds when $\KK$ has positive characteristic.

	Unfortunately, Hahn series fail to provide full basis of solutions of any Mahler equations. The equations $\puip(y)=cy$ with $c \in \KK^\times \setminus\{1\}$ or $\puip(y)=y+1$ have no solutions in the field of Hahn series, as one can check by looking at the hypothetical valuation of such a solution. However, it is possible to find solutions to these equations in some formal ring extension of the field of Hahn series. We denote by $e_c$ a solution to the former and by $\logm$ a solution to the latter, in such an extension. For convenience, we also let $e_1=1$. As mentioned in \cite[Section~2.2]{FP25}, one can construct such a ring extension with field of constants $\{f : \varphi_p(f)=f\}=\KK$. Then, it follows from \cite{FR25} that any Mahler equation over $\K[[z]]$ has a basis of solutions\footnote{A basis of solutions of \eqref{eq:Mahler} is a set of $m$ solutions which are linearly independent over the field of constants $\KK$.} $y_1,\ldots,y_m$ of the form 
	\begin{equation}\label{eq:basis_solutions}
	y_i = \sum_{c \in \KK^\times} \sum_{j \in \Z_{\geq 0}} \sum_{\omeg \in \Omega} f_{i,c,j,\omeg}(z)\bxi_{\omeg} e_c \logm^j
	\end{equation}
	where the $f_{i,c,j,\omeg}$ are Puiseux series with coefficients in $\KK$ and each sum has finite support. When $\K$ is of characteristic $0$, this decomposition is unique if one restricts $\omeg$ to the set $\Omega_0 \subset \Omega$ of tuples $\omeg=(\u,\a)$ for which $\u$ is of the form \eqref{eq:basis_rec_lin} and the entries of $\a$ have denominators coprime with $p$ and numerators not divisible by $p$, see \cite{FR25}. The aim of this article is to provide an algorithm to compute such a basis of solutions.
	\begin{thm}\label{thm:algo_solutions_equations}
		Let $\K$ be an effective field of characteristic $0$.
		Algorithm \ref{algo:basis_solutions} takes as input a Mahler equation\footnote{We assume that we are provided with a way to compute as many coefficients as needed in the power series expansion of $a_0,\ldots,a_m$.} over $\K[[z]]$ of the form \eqref{eq:Mahler} and returns a basis $y_1,\ldots,y_m$ of solutions of \eqref{eq:Mahler} under the form \eqref{eq:basis_solutions}. Precisely, it returns:
		\begin{itemize}
			\item a finite set $K_0 \subset \KK^\times$, an integer $j_0 \geq 0$, and a finite set $\Omega_1 \subset \Omega_0$ such that the supports of the sums \eqref{eq:basis_solutions} are included in $K_0 \times\{0,\ldots,j_0\}\times \Omega_1$,
			\item an integer $d$ and an integer $v$ such that $f_{i,c,j,\omeg}(z) \in z^{-v}\KK[[z^{1/d}]]$ for every $i,c,j,\omeg$,
			\item the coefficients of the Puiseux expansion of each $f_{i,c,j,\omeg}(z)$ up to any prescribed order.
		\end{itemize}
	\end{thm}

	\begin{rem*}
		\label{rem:computation_puiseux_equations}\begin{enumerate}[label={\rm (\roman*)}]
			\item
			The nonzero Puiseux series $f_{i,c,j,\omeg}(z)$ appearing in \eqref{eq:basis_solutions} are also solutions of some $p$-Mahler equations. Such Puiseux series are uniquely determined by their associated equation and their first coefficients. It is possible to adapt Algorithm \ref{algo:basis_solutions} to associate to each $f_{i,c,j,\omeg}(z)$ a $p$-Mahler equation it satisfies. In particular, this allows one to check whether such a series is $0$ or whether there are linear relations, or algebraic relations of a given degree, between them (see \cite{AF17}).
			\item
			Our algorithm can be used to compute solutions in specific rings, such as power, Puiseux or Hahn series. While it does not improve the existing algorithm for Puiseux series, it significantly improves the one for Hahn series. Indeed, the algorithm in \cite{FR24} computes the coefficients of $z^{\gamma}$ of any Hahn series solution for exponents $\gamma$ inside a given finite set. Using Theorem \ref{thm:algo_solutions_equations}, given an integer $N$, we can provide a closed formula for the coefficients of $z^\gamma$ with $\vert \gamma \vert \leq N$, which is an infinite set.
		\end{enumerate}
	\end{rem*}

	\subsection{Linear Mahler systems}\label{sec:systems}
	Recall that we let $\puip$ denote the map $z \mapsto z^p$. 
	Although Theorem \ref{thm:algo_solutions_equations} is about linear equations, the algorithm it refers to relies on our ability to solve linear systems.  A $p$-Mahler system over $\K((z))$ is a system of the form
	\begin{equation}
	\label{eq:Mahler_sys}
	\puip(Y)=AY
	\end{equation}
	with $A \in \mathrm{GL}_m(\K((z)))$. Let $\Puis_\K = \bigcup_{d\geq 1}\K((z^{1/d}))$ and $\Puis_\KK = \bigcup_{d\geq 1}\KK((z^{1/d}))$ denote the fields of Puiseux series with coefficients in $\K$ and $\KK$ respectively. Let $\Hahn_{\KK}=\KK((z^{\QQ}))$ denote the field of Hahn series with value group $\QQ$ and coefficients in $\KK$, that is the set of formal series $\sum_{\gamma \in \QQ} a_\gamma z^\gamma$ with $a_\gamma \in \KK$ such that $\{\gamma\,:\,a_\gamma \neq 0\}$ is a well-ordered subset of $\QQ$ with respect to the restriction of the order $<$. The map $\puip$ extends to $\Puis_\K$, $\Puis_\KK$ and $\Hahn_{\KK}$ in the obvious way. Roques and the first author \cite[Theorem~17]{FR25} proved that any $p$-Mahler system \eqref{eq:Mahler_sys} over $\K((z))$ has a \textit{fundamental matrix of solutions} of the form
	\begin{equation}\label{eq:mat_fund_sol}
	PHe_C
	\end{equation}
	where 
	\begin{enumerate}[label = \rm (\alph*)]
		\item\label{item:Puiseux} the matrix $P \in \GL_m(\Puis_\KK)$ is such that $\Theta:=\puip(P)^{-1}AP$ is block upper triangular with entries in $\KK[z^{-1/*}]$ and constant non-singular diagonal blocks;
		\item\label{item:Hahn} the matrix $H$ 
		\begin{itemize}
			\item belongs to $\GL_m(\Hahn_\KK)$ and satisfies $\puip(H)C=\Theta H$, where $C \in\GL_m(\KK)$ is the constant term of $\Theta$,
			\item is block upper triangular with blocks fitting with the block decomposition of $\Theta$ and diagonal blocks equal to the identity;
			\item has coefficients which are $\KK[z^{\pm 1/d}]$-linear combinations of some $\bxi_{\omega}$, $\omeg \in \Omega$, where $d \in \Z_{>0}$;
		\end{itemize}
		\item\label{item:const} the matrix $e_C$ has coefficients which are $\KK$-linear combinations of the $e_c\logm^j$, $c\in \KK^\times$, $j\in \Z_{\geq 0}$, and satisfies $\puip(e_C)=Ce_c$.
	\end{enumerate}
	We make the computation of such a fundamental matrix of solutions effective.
	\begin{thm}
		\label{thm:algo_solutions_systems}
		Let $\K$ be an effective field. Consider a $p$-Mahler system \eqref{eq:Mahler_sys} over $\K((z))$. Algorithms \ref{algo:solutions_blocks}, \ref{algo:matrixH} and \ref{algo:constante} provide a fundamental matrix of solutions of the form \eqref{eq:mat_fund_sol}. More precisely:
		\begin{itemize}
			\item Algorithm \ref{algo:solutions_blocks} provides a solution to \ref{item:Puiseux} over $\K$, that is a matrix with the required form $\Theta \in \GL_m(\K[z^{-1/*}])$ and a truncation up to any order of a matrix $P \in \GL_m(\Puis_\K)$ of Puiseux series with coefficients in $\K$ such that $\Theta=\puip(P)^{-1}AP$;
			\item Algorithm \ref{algo:matrixH} provides a solution to \ref{item:Hahn}, that is a description of the entries of $H$ as $\KK$-linear combinations of some $\bxi_\omeg$, $\omeg \in \Omega$;
			\item Algorithm \ref{algo:constante} provides a solution to \ref{item:const}, that is a description of the entries of $e_C$ as linear combinations of some $e_c\logm^j$, $c\in \KK^\times$, $j \in \Z_{\geq 0}$.
		\end{itemize}
	\end{thm}

	\begin{rem*}
		\begin{itemize}
			\item It follows from Theorem \ref{thm:algo_solutions_systems} that a solution to \ref{item:Puiseux} exists with $\K$ as a base field instead of $\KK$. Such a fact does not follow from \cite{FR25} directly but is the consequence of the construction in Section~\ref{sec:algo_blocs}. In contrast, it is not always possible to perform \ref{item:const} on $\K$ (see Section~\ref{sec:alg_clo}).
			\item The entries of the matrix $H$ produced by Algorithm \ref{algo:matrixH} are linear combinations of some $\bxi_\omeg$ over the field $\KK$, rather than over the ring $\KK[z^{\pm 1/d}]$ as was the case in \ref{item:Hahn}. This refines the results from \cite{FR25}.
		\end{itemize}
	\end{rem*}
	To any linear equation \eqref{eq:Mahler} one can associate a system by considering its companion matrix. If one obtains a fundamental matrix of solutions of this system thanks to Theorem \ref{thm:algo_solutions_systems}, then one gets a solution as in Theorem \ref{thm:algo_solutions_equations} by looking at the first row of $PHe_C$.

	\subsection{Organization of the paper} In Section~\ref{sec:motivations} we discuss connections between Mahler equations and other areas of research. This section is independent of the remainder of the paper and may be skipped on a first reading. 
The main contribution of this paper is Algorithm~\ref{algo:solutions_blocks} whose proof occupies the next two sections.
In Section~\ref{sec:preliminaires} we gather preliminary results about the matrices $P$ and $\Theta$ introduced in Section~\ref{sec:systems} that will reduce the problem to a finite-dimensional one over $\K$.
In Section~\ref{sec:algo_blocs}, we present and prove the correctness of Algorithm \ref{algo:solutions_blocks}. Once one knows $\Theta$, in Section~\ref{sec:HeC} we describe Algorithms \ref{algo:matrixH} and \ref{algo:constante} which allow us to compute the matrices $H$ and $e_C$, respectively, introduced in Section~\ref{sec:systems}. They provide algorithmic versions of the constructions from \cite{FR25} and \cite{Ro18}. Then, we state Algorithm~\ref{algo:basis_solutions} which is the core of Theorem \ref{thm:algo_solutions_equations}. Eventually, in Section~\ref{sec:RudinShapiro}, we run this algorithm on two examples.

	\section{Applications and connections}\label{sec:motivations}
	
	The study of Mahler equations is motivated by their occurrence in diverse areas of mathematics. The algorithms developed in this paper provide computational tools relevant to each of these domains. We outline some key connections in what follows. This material is not needed for the rest of the paper. Thus, the reader should feel free to proceed directly to Section~\ref{sec:preliminaires}.
	
	\subsection*{Automatic sequences and automatic real numbers} 
	
	Mahler equations naturally occur in the study of finite automata. A deterministic finite automaton can be used as a transducer to produce a sequence $(u_n)_{n\in \N}$ where $u_n$ denotes the output of the automaton when it reads the expansion of $n$ in base $p$ (see \cite{AS03} for details). When the elements $u_n$ belong to some field $\K$, a classical result states that the power series $\sum_n u_n z^n$ satisfies some linear Mahler equation with coefficients in $\K[z]$ (see \cite{Cob68}). When $\K:=\mathbb Q$, combined with Mahler's method, this connection has been used to establish the following results\footnote{The first result was originally proved using a different approach, see \cite{AB07}.}:
	\begin{itemize}
		\item the base-$b$ expansion of an irrational algebraic real number cannot be produced by a finite automaton \cite{PPH15,AF17};
		\item the expansions in two multiplicatively independent bases of an irrational real number cannot both be produced by automata \cite{AF25}.
	\end{itemize}
	It has also been used to reprove and generalize the well-known Cobham's theorem \cite{SS19, AF25}.
	
	\subsection*{Regular sequences}
	Regular sequences are a generalization of sequences produced by a finite automaton. Strictly speaking, the $p$-regular sequences are the sequences that may produce a weighted finite automaton when reading the integers written in base $p$. Among them one finds the sequences whose $n$th term is the sum of digits of $n$ in base $p$, the $p$-adic valuation of $n$ or of $n!$, the complexity of the merge-sort algorithm in a set with $n$ elements, the number of odd entries in the $n$th row of Pascal's triangle, the $n$th Cantor number. As established by Becker \cite{Bec94}, the generating series of a $p$-regular sequence is solution to some $p$-Mahler equation.

	\subsection*{Algebraic power series in $\mathbb F_p[[z]]$} 
	
	When $p$ is prime, the elements of $\mathbb F_p[[z]]$ that are algebraic over $\mathbb F_p[z]$ are exactly the solutions of linear Mahler equations over $\mathbb F_p[z]$. Actually, given an algebraic equation over $\mathbb F_p[z]$, deriving a $p$-Mahler equation whose solutions in $\mathbb F_p[[z]]$ are the power series solutions of the former algebraic equation is straightforward, using the identity $f(z^p)=f(z)^p$ in $\mathbb F_p[[z]]$. For example, when $p=2$, any solution to
	$$
	c_0+c_1y(z)+c_2y(z)^2=0,
	$$
    with $c_0,c_1,c_2\in \mathbb F_p[z]$ and $c_0c_1c_2\neq 0$, satisfies the equation 
    $$c_1c_0y(z)+(c_1^2-c_2c_0)y(z)^2-c_2^2y(z)^4=0$$
    which, for power series in $\mathbb F_2[[z]]$, is equivalent to the Mahler equation
	$$
	c_1c_0y(z)+(c_1^2-c_2c_0)y(z^2)-c_2^2y(z^{2^2})=0.
	$$
	One should note, however, that this does not imply that our main algorithm for Mahler equations provides all solutions of the initial algebraic equation. Indeed, the identity $f(z^p)=f(z)^p$ does not extend to $\Hahn_{\overline{\mathbb F_p}}[(e_c)_{c},\logm]$.

	\subsection*{Periods of Drinfeld modules} 
	
	The theory of Drinfeld modules in positive characteristic provides analogs of numerous famous periods. As shown by Denis \cite{Den00}, some of them arise as values of solutions of Mahler equations. For example, a $p$-analog of $\pi$ is given by $\prod_{n\geq 0}(1-\theta^{1-p^n})$ in the completion $\K$ of the algebraic closure of $\mathbb F_p((\theta^{-1}))$. This is the value at $\theta^{-1}$ of the function $f(z):=\prod_{n\geq 0} (1-\theta z^{p^n})\in \K[[z]]$ which satisfies the $p$-Mahler equation
	$$
	f(z)-(1-\theta z)f(z^p)=0\,.
	$$
	Another example is the analog $\zeta_C$ of the classical Riemann $\zeta$ function. Carlitz \cite{Car35} proved that $\zeta_C(s)=f_s(\theta)$ for $s \in \{1,\ldots,p-1\}$, where $f_s(z) \in \K[[z]]$ satisfies the inhomogeneous $p$-Mahler equation
	$$
	f_s(z^p) - (-1)^s (z^p - \theta)^s f_s(z) = (-1)^s (z^p - \theta)^s\,.
	$$
	Using this equation, Dennis \cite{Den06} was able to prove the algebraic independence of $\zeta_C(1),\ldots,\zeta_{C}(p-1)$ over $\mathbb F_p(\theta)$ (see also \cite{Fer18}).

	\subsection*{Purity results} 
	When considering the minimal differential equation satisfied by some $E$-function, Andr\'e \cite{An00} proved that the functions $f(z)$ appearing in \eqref{eq:diff} are $E$-functions as well. Surprisingly, such a ``purity theorem'', which does not concern the values of these functions, implies the celebrated Siegel--Shidlovskii theorem on algebraic relations between values of $E$-functions at algebraic points and even some remarkable refinements due to Beukers \cite{Be06}. A purity theorem analogous to Andr\'e's has recently been established in the framework of Mahler equations \cite{FR25} and similarly shares a deep connection with results concerning the values of solutions of these equations \cite{ABS23}. One instance of this result is the following: if \eqref{eq:Mahler} is the minimal Mahler equation satisfied by some $p$-regular power series, then it has a basis of solutions of the form \eqref{eq:basis_solutions} for which each of the series $f_{i,c,j,\omeg}$ is $p$-regular too.
	
	\subsection*{Galois theory of Mahler equations} 
	
	Recent results have established that any algebraic relation between values of solutions of Mahler equations at some nonzero algebraic point has a functional origin (see \cite{AF24}). 
    The study of functional relations is a delicate task and constitutes the subject of difference Galois theory. The Galois group of a Mahler system \eqref{eq:Mahler_sys} is a linear algebraic group which encodes the relations between the entries of any fundamental matrix of solutions. From this point of view, being able to compute a fundamental matrix of solutions of a Mahler system as in Theorem \ref{thm:algo_solutions_systems} should be of some help to understand the algebraic relations between solutions of Mahler equations.

	\section{Properties of the Puiseux part of a fundamental matrix of solutions}\label{sec:preliminaires}
	We fix an integer $p\geq 2$, a field $\K$ and let $\KK$ denote an algebraic closure of $\K$. 
	The general form \eqref{eq:mat_fund_sol} of a fundamental matrix of solutions leads us to look for pairs $(P,\Theta)$ satisfying Condition \ref{item:Puiseux}. This section is devoted to the study of such pairs.
	\begin{defi}\label{def:admissible}
		We say that a pair of $m\times m$ matrices $(P,\Theta)$ is admissible with respect to \eqref{eq:Mahler_sys} if the following holds: 
		\begin{itemize}
			\item $\Theta$ is block upper triangular with entries in $\K[z^{-1/*}]$ and has non-singular constant diagonal blocks;
			\item $P \in \GL_{m}(\Puis_{\K} )$ is such that 
			\begin{equation}\label{eq:F1Theta=AF1}
			\puip(P)\Theta=AP\, .
			\end{equation}
		\end{itemize}
	\end{defi}

	\subsection{Equations, systems and modules}\label{sec:modules}
	We briefly recall the vocabulary of equations, systems and modules. 
	
	Let $\Hahn_\K$ denote the field of Hahn series with coefficients in $\K$ and value group $\QQ$. An element of $\Hahn_\K$ is a formal series $f(z)=\sum_{\gamma \in \QQ}f_\gamma z^\gamma$ with $f_\gamma \in \K$ such that the set $\{\gamma : f_\gamma \neq 0\}$ is a well-ordered subset of $\QQ$ for the total order $<$. This guarantees that $\Hahn_\K$ is a field extension of $\K((z))$ (see \cite{Ro18} for details). 
	
	Given a subfield $\mathbb F$ of $\Hahn_\K$, we say that two $p$-Mahler systems $\puip(Y)=AY$ and $\puip(Y)=BY$ are $\mathbb F$-equivalent if there exists a matrix $F \in \GL_m(\mathbb F)$ such that $A=\puip(F)BF^{-1}$. 
    When $\K$ is algebraically closed, by the \textit{cyclic vector lemma}\footnote{The proof in \cite{FP22} is written over $\Q(z)$ but it works over $\mathbb F$ since its field of coefficients $\K$ is algebraically closed, hence, infinite.} \cite{FP22}, any $p$-Mahler system 
	\begin{equation}\label{eq:Mahler_system_dF}
	\puip(Y)=AY,\quad A(z) \in \GL_m(\mathbb F)
	\end{equation}
	is $\mathbb F$-equivalent to a $p$-Mahler system whose matrix is of the form
	\begin{equation}\label{eq:companion}
	B=\begin{pmatrix}
	0 & 1 & 0 & \cdots & 0
	\\ & \ddots& \ddots & &
	\\ & & \ddots & \ddots & 
	\\ & & & 0 & 1
	\\ a_0& \cdots & \cdots & \cdots & a_{m-1}
	\end{pmatrix}
	\end{equation} 
	where $a_0,\ldots,a_{m-1} \in \mathbb F$, $a_0a_{m-1}\neq 0$. A vector $Y$ is solution to the $p$-Mahler system with matrix $B$ if and only if $Y=\prescript{\rm t}{}(y,\puip(y),\ldots,\puip^{m-1}(y))$ for some $y$ solution to the following $p$-Mahler equation
	\begin{equation}\label{eq:Mahler_eq}
	a_0y+a_1\puip(y)+\cdots+a_{m-1}\puip^{m-1}(y)-\puip^m(y)=0\,.
	\end{equation}
	
	Consider the ring $\mathscr{D}_{\mathbb F}:=\mathbb F\langle\Phi\rangle$ of non-commutative polynomials in the indeterminate $\Phi$ with the property that $\Phi f=\puip(f)\Phi$ for any $f \in \mathbb F$. To any companion matrix $B$ as in \eqref{eq:companion} we associate the operator $L_{B}=a_0+a_1\Phi+\cdots+a_{m-1}\Phi^{m-1}-\Phi^m$. Then \eqref{eq:Mahler_eq} may be rewritten as $L_B(y)=0$.
	
	To any system \eqref{eq:Mahler_system_dF}, we can associate a $\mathscr{D}_{\mathbb F}$-module $M_A$ of finite rank as follows: the underlying $\mathbb F$-vector space is $\mathbb F^m$ and 
	$\Phi$ acts on $\mathbb F^m$ by
	$$
	\forall \v \in \mathbb F^m, \quad \Phi(\v)=A^{-1}\puip(\v) \,.
	$$
	Conversely, to any $\mathscr{D}_{\mathbb F}$-module $M$ of finite rank we can associate a Mahler system by choosing a basis. It is not difficult to prove that if $\puip(Y)=AY$ and $\puip(Y)=BY$ are $\mathbb F$-equivalent then the  $\mathscr{D}_{\mathbb F}$-modules $M_A$ and $M_B$ are isomorphic. Furthermore, if $B$ is a companion matrix, $M_B$ is isomorphic to the module $\mathscr{D}_{\mathbb F}/\mathscr{D}_{\mathbb F}L_B$.
	
	To any operator $L=a_0+a_1\Phi +\cdots +a_m\Phi^m \in \mathscr{D}_{\mathbb F}$ we associate a \textit{Newton polygon} $\mathcal N(L)$, which is the convex hull of the set of points $\left\{(p^i,j)\,:\, 0 \leq i \leq m,\,j\geq \val a_i\right\}$. 
	We let $\mu_1,\ldots,\mu_r \in \mathbb Q$ denote the finite slopes of $\mathcal N(L)$. Let $d(L)$ be the least integer $d$ such that $\mu_1,\ldots,\mu_r \in d^{-1}\Z[p^{-1}]$. This is the least integer $d$ relatively prime with $p$ and such that the denominators of $\mu_1,\ldots,\mu_r$ divide $dp^k$ for some $k$.

	\subsection{Ramification order of an admissible pair} Recall that we assume $\K$ to be algebraically closed. We now take $\mathbb F:=\K((z))$ and consider a Mahler system
	\begin{equation}\label{eq:Mahler_system_d}
	\puip(Y)=AY,\quad A(z) \in \GL_m(\K((z)))\, .
	\end{equation}
	Then, the module $M_A$ is isomorphic to some $\mathscr{D}_{\K((z))}$-module  $\mathscr{D}_{\K((z))}/\mathscr{D}_{\K((z))}L$. We set $d(A):=d(L)$. It can be proved that this does not depend on the choices of $L$ such that $M_A \simeq \mathscr{D}_{\K((z))}/\mathscr{D}_{\K((z))}L$, a fact we do not need here. Note that $d(A(z^{d(A)}))=1$ for any matrix $A$, so we can always reduce to the case where $d(A)=1$. Note also that $d(A) \mid p^m-1$ so that $d(A(z^{p^m-1}))=1$.

	\begin{prop}
		\label{prop:ramification} Suppose that $\K$ is algebraically closed. Then, the system \eqref{eq:Mahler_system_d} admits an admissible pair $(P,\Theta)$ such that $P$ and $\Theta$ have entries in $\K((z^{1/d(A)}))$.
	\end{prop}
	As mentioned in the introduction, we will obtain as a by-product of Algorithm~\ref{algo:solutions_blocks} that this proposition still holds when $\K$ is not algebraically closed.
	
	\begin{proof}
		Up to replacing $z$ by $z^{d(A)}$, we may assume that $d(A)=1$. Then, Proposition~\ref{prop:ramification} can be reformulated as follows: \textit{the system \eqref{eq:Mahler_system_d} is $\K((z))$-equivalent to a block upper triangular system with matrix in $\GL_m(\K[z^{-1}])$ and constant diagonal blocks}.
		Since we closely follow the proof of Step~1 in \cite[Theorem~17]{FR25}, we shall be brief and present only the necessary adaptations. The existence of an admissible pair follows from a construction in two parts. The first one is a factorization of Mahler operators which implies that any Mahler system is Puiseux-equivalent to an upper triangular system with constant diagonal coefficients and Puiseux above-diagonal coefficients. The second part is a construction which implies that such a system is Puiseux-equivalent to a block upper triangular system with constant diagonal blocks and whose above-diagonal blocks have entries in $\K[z^{-1}]$. We only have to explain that these two steps may be carried out over $\K((z))$ instead of the field of Puiseux series and that $\K$ may be any algebraically closed field.

		\vskip 5 pt 
		\noindent \textit{Part 1: Factorization of Mahler operators.}
		We prove the following claim: \textit{if $L \in \mathscr{D}_{\K((z))}$ is such that $d(L)=1$, then one has $L=L_1\cdots L_m$ with $L_i \in \mathscr{D}_{\K((z))}$ of degree $1$ in $\Phi$}. 
		
		The proof is an adaptation of the one of \cite[Proposition 15]{Ro24}. By induction, it is enough to prove that $L$ admits a factorization $L=MN$, with $M,N \in \mathscr{D}_{\K((z))}$, $N$ of degree $1$ in $\Phi$ and $d(M)=1$ (then, apply the induction hypothesis to $M$). The first slope $\mu_1$ of $\mathcal N(L)$ is of the form $\mu_1=a/(p^k-1)$ for some $a \in \Z$ and $k \geq 1$. Since $d(L)=1$, $\mu_1$ has to be an integer. Without loss of generality, we may suppose that $\mu_1=0$. Indeed, with the notation of \cite{Ro24}, this amounts in replacing $L$ with $L^{[\theta_{\mu_1}]}$, which does not modify $d(L)$ since $\mu_1 \in \Z$ (\cite[Lemma 17]{Ro24}). By \cite[Lemma 20]{Ro24}, there exists $f \in \Hahn_\K$ such that $f(0)=1$ and $L(fe_c)=0$ for some $c\in\K^\times$. The proof given there readily extends to any algebraically closed field $\K$. Since $L$ has coefficients in $\K((z))$, $f$ actually belongs to $\K((z))$ (see \cite[Lemma 28]{FP25}). Let $M$ denote the result of the right Euclidian division of $L$ by $N:=(\Phi-c)f^{-1}$ in $\mathscr{D}_{\K((z))}$. Then $M \in \mathscr{D}_{\K((z))}$ and $L=MN$. By \cite[Lemma 21]{Ro24}, each slope of $M$ is $p$ times a slope of $L$. \textit{A fortiori}, $d(M)=1$. This proves the claim.
		
		\vskip 5 pt 
		\noindent \textit{Part 2: Reduction to $\K[z^{-1}]$.}
		Let $L \in \mathscr{D}_{\K((z))}$ be such that $M_A\simeq \mathscr{D}_{\K((z))}/\mathscr{D}_{\K((z))}L$. From the first step, there exists $L_1,\ldots,L_m \in \mathscr{D}_{\K((z))}$ of order $1$ such that
		$$
		L=L_1\cdots L_m\,.
		$$
		Arguing as in \cite[Section~4.1.1]{FR25}, we conclude that the Mahler system \eqref{eq:Mahler_system_d} is $\K((z))$-equivalent to an upper triangular system with constant diagonal coefficients and above-diagonal coefficients in $\K((z))$. Then we follow the construction in \cite[Section 4.1.2]{FR25} without any modification. The key argument is that \cite[Lemma~20]{FR25} remains valid when one replaces the field of Puiseux series with $\K((z))$ and when $\K$ is any field. The remainder of the construction proceeds identically.
	\end{proof}

	\subsection{Valuations of an admissible pair}
	The aim of this section is to bound from below $\val P$ and $\val \Theta$ for some admissible pair $(P,\Theta)$. We consider the following set
	$$
	\mathcal S_p:=\{s \in \Z_{<0}\,:\, p \nmid s\} \cup\{0\}
	$$
	of non-positive integers either equal to $0$ or not divisible  by $p$.
	
	\begin{prop}
		\label{prop:incomplete_basis}
		Suppose that $\K$ is algebraically closed and consider a Mahler system \eqref{eq:Mahler_system_d} with $d(A)=1$. Then, there exists an admissible pair $(P,\Theta)$ for which
		\begin{enumerate}[label = $\rm (\roman*)$]
			\item \label{item:support_Theta} $P \in \GL_m(\K((z)))$ and $\Theta$ has support in $\mathcal S_p$;
			\item \label{item:valuation} one has
			\begin{equation*}\label{eq:valuation}
			\val P \geq \frac{\val A}{p-1},\quad \text{and}\quad \val \Theta\geq \frac{pm \val A - p\val(\det A)}{p-1}\,.
			\end{equation*}
		\end{enumerate}
		Furthermore, each admissible pair satisfying \ref{item:support_Theta} also satisfies \ref{item:valuation}.
	\end{prop}
	The proof is divided in two steps, corresponding respectively to \ref{item:support_Theta} and \ref{item:valuation}.
	\subsubsection{Proof of \ref{item:support_Theta} of Proposition \ref{prop:incomplete_basis}}
	
	By Proposition \ref{prop:ramification}, the system has an admissible pair $(P,\Theta)$ with $P \in \GL_m(\K((z)))$ and $\Theta \in \GL_m(\K[z^{-1}])$. Since $\K$ is algebraically closed, there exists such pair with $\Theta$ upper triangular. For such pairs we let $\theta_1,\ldots,\theta_m$ denote the diagonal entries of $\Theta$ and $\theta_{i,j}(z)$, $i<j$, denote its above diagonal entries.
	
	Let $\mathcal I:=\{(i,j) \in \{1,\ldots,m\}^2\,:\,i<j\}$. We define an order on $\mathcal I$ as follows: $(i,j)\prec(k,l)$ if either $j<l$ or ($j=l$ and $i>k$). To scan the above-diagonal entries of $\Theta(z)$ with respect to this order one proceeds from left-to-right and then from bottom-to-top. 
	
	Suppose by contradiction that there exists no admissible pairs $(P,\Theta)$ for which $P \in \GL_m(\K((z)))$ and $\Theta$ is uper triangular with support is in $\mathcal S_p$. Thus, for any admissible pair $(P,\Theta)$  with $P \in \GL_m(\K((z)))$ and $\Theta$ upper triangular, there exists $(i,j)\in \mathcal I$  such that $\theta_{i,j}(z)$ has a non-zero multiple of $p$ in its support: we let $(i_\Theta,j_\Theta)$ be the least of them with respect to $\prec$. 
	Consider the pairs $(P,\Theta)$ such that $(i_\Theta,j_\Theta)$ takes its maximal value with respect to $\prec$ and let $(i_0,j_0)$ denote this value. Recall that, for any $(P,\Theta)$ such that $(i_\Theta,j_\Theta)=(i_0,j_0)$, the support of $\Theta_{i,j}(z)$ for each $(i,j)\prec (i_0,j_0)$ is included in $\mathcal S_p$. Furthermore, by assumption, for any pair $(P,\Theta)$ such that $(i_\Theta,j_\Theta)=(i_0,j_0)$, there exists $\nu \in \Z_{>0}$ such that $z^{-\nu p}$ belongs to the support of $\theta_{i_0,j_0}(z)\in \K[z^{-1}]$. When the pair $(P,\Theta)$ is fixed, we let $\nu_\Theta$ denote the greatest such $\nu$. Then, we let $\nu_0$ denote the minimum of the integers $\nu_\Theta$, among all the pairs $(P,\Theta)$ for which $(i_\Theta,j_\Theta)=(i_0,j_0)$. From now on we fix a pair $(P,\Theta)$ with $(i_\Theta,j_\Theta)=(i_0,j_0)$ and $\nu_\Theta=\nu_0$. We shall obtain a contradiction.
	
	Let $\eta$ denote the coefficient of $z^{-\nu_0p}$ in $\theta_{i_0,j_0}(z)$. Let $M$ denote the matrix which is the identity but for its $(i_0,j_0)$th entry which is equal to $\eta\theta_{j_0}^{-1}z^{-\nu_0}$. Set 
	$$
	\widetilde{\Theta}:=\puip(M)^{-1}\Theta M.
	$$
	This matrix is upper triangular, with diagonal coefficients $\theta_1,\ldots,\theta_m$ and above diagonal coefficients in $\K[z^{-1}]$. Furthermore, when $(i,j)\prec(i_0,j_0)$ the $(i,j)$th coefficient of $\widetilde{\Theta}$ is the same as the one of $\Theta$ and the $(i_0,j_0)$th coefficient of $\widetilde{\Theta}$ is equal to
	\begin{equation}\label{eq:tilde_Theta}
	\widetilde{\theta}_{i_0,j_0}(z):=\theta_{i_0,j_0}(z) - \eta z^{-\nu_0p} +\eta\theta_{i_0}\theta_{j_0}^{-1}z^{-\nu_0}\,.
	\end{equation}
	Last, setting $\widetilde{P}=PM$ we have $\puip(\widetilde{P})\widetilde{\Theta}=A\widetilde{P}$ so that the pair $(\widetilde{P},\widetilde{\Theta})$ is admissible. By maximality of $(i_0,j_0)$, the support of $\widetilde{\theta}_{i_0,j_0}(z)$ has an element of the form $-\mu p$, for some $\mu \in \Z_{\geq 1}$. By minimality of $\nu_0$, $\mu \geq \nu_0$. This contradicts \eqref{eq:tilde_Theta}.
	\qed
	
	\subsubsection{Proof of \ref{item:valuation} of Proposition \ref{prop:incomplete_basis}}  
	Let $(P,\Theta)$ be an admissible pair satisfying \ref{item:support_Theta}. Up to conjugation we may and will assume that $\Theta$ is upper triangular. We start by establishing the lower bound on $\val P$. Let $P_1,\ldots,P_m$ enumerate the columns of $P$ from left to right. We proceed by contradiction. Suppose that the lower bound for $\val P$ does not hold and let $i$ be the least integer such that $\val P_i<\val A/(p-1)$. We let 
	$$\theta_{1}(z),\ldots,\theta_{i-1}(z),\lambda_i,0,\ldots,0$$
	denote the entries of the $i$th column of $\Theta(z)$. Then, we infer from \eqref{eq:F1Theta=AF1} that
	$$
	\lambda_i\puip(P_i) = AP_i - \theta_1(z)\puip(P_1) - \cdots - \theta_{i-1}(z)\puip(P_{i-1})
	$$
	Since $\val P_i<\val A/(p-1)$, we have $\val(AP_i)\geq \val(A) + \val(P_i)> p\val(P_i) =\val(\lambda_i\puip(P_i))$. Thus, the valuation of $\puip(P_i)$ is equal to the valuation of
	$$
	\theta_1(z)\puip(P_1) + \cdots + \theta_{i-1}(z)\puip(P_{i-1})\,.
	$$
	Such a valuation is of the form $\nu + p\mu$, where $\nu$ is in the support of some $\theta_{j}(z)$, $1\leq j<i$, and $\mu \geq \min_{j<i}\val(P_j)$. Since $\val(\puip(P_i)) = p\val(P_i)$, $p$ must divide $\nu$. Since $\nu \in \mathcal S_p$, we have $\nu =0$. Thus, $\val P_i =\mu \geq \min_{j<i}\val(P_j) \geq \val A/(p-1)$, a contradiction.
	
	To establish the bound for $\val \Theta$ we first bound $\val P^{-1}$ from below. Since $\Theta$ is upper triangular with invertible constant diagonal entries, $\det \Theta \in \K^\times$. Hence, $\val \det \Theta=0$ and we infer from \eqref{eq:F1Theta=AF1} that
	\begin{equation}\label{eq:val_detP}
	\val(\det P) = \frac{\val(\det A)}{p-1}\,.
	\end{equation}
	From the formula $P^{-1}=(\det P)^{-1}\prescript{\rm t}{}{\operatorname{Com}P}$, we deduce that
	\begin{multline*}
	\val P^{-1} = \val(\operatorname{Com} P) - \val(\det P) 
	\\ \geq (m-1) \val P - \frac{\val (\det A)}{p-1}\geq \frac{(m-1)\val A-\val( \det A)}{p-1}.
	\end{multline*}
	Last, since $\Theta=\puip(P^{-1})AP$, we have
	$$
	\val \Theta \geq p\val P^{-1} + \val A + \val P \geq \frac{pm\val A-p\val(\det A)}{p-1}
	$$
	\qed
	
	\subsection{Laurent series expansion of $P$}
	We fix a Mahler system \eqref{eq:Mahler_system_d} for which $d(A)=1$. We do not necessarily assume that $\K$ is algebraically closed. Note that if $(P,\Theta)$ is admissible it remains so when considering the equation over $\KK((z))$. In particular, the conclusion of Proposition \ref{prop:incomplete_basis} remains true. In the remainder of the paper we shall take the following notation:
	\begin{align}\label{eq:nota_val}
	\nonumber&\nu_\Theta  :=\min\left\{ v\geq   \frac{pm \val A - p\val(\det A)}{p-1}\,:\, v\in \mathcal S_p \right\},\qquad  \nu_P := \left\lceil \frac{\val A}{p-1}\right\rceil ,
	\\  &\nu:= \min\{\nu_P,\,p\nu_P + \vb\}+\nu_\Theta,
	\\ \nonumber &\mu := \max\left\{\left\lceil -\frac{\vb + \nu_\Theta}{p-1}\right\rceil,\,\frac{\val \det(A)}{p-1}-(m-1)\nu_P\right\}. 
	\end{align}
	The integers $\nu_P$ and $\nu_\Theta$ should be considered as lower bounds for the valuations of the components of an admissible pair $(P,\Theta)$. Since $\val(\det A)\geq m \val A$ we have $\nu_{\Theta}\leq 0$. Last, it follows from \eqref{eq:val_detP} that $\mu$ is an integer.

	Given an admissible pair $(P,\Theta)$ --if such a pair exists-- one can group the columns of $P$ according to the block decomposition of $\Theta$. The aim of this section is to establish recurrence formulas for the Laurent series expansion of these groups of columns.

	\subsubsection{A recurrence formula for the blocks of columns}
	\label{subsec:recurrBlocks}
	Let $P \in \GL_m(\K((z)))$ and let $\Theta \in \GL_m(\K[z^{-1}])$ be block upper triangular with blocks of sizes $b_1,\ldots,b_s$, respectively, support in $\mathcal S_p$, $\val(P)\geq \nu_P$ and $\val \Theta \geq \nu_\Theta$. We do not assume at this stage that $(P,\Theta)$ is an admissible pair. We partition the columns of $P$ according to the block decomposition of $\Theta$ and denote by $Q_1,\ldots,Q_s$ the corresponding matrices. Precisely, $Q_j$ is the $m\times b_j$ matrix consisting of the columns $b_1+\cdots+b_{j-1}+1$ to $b_1+\cdots+b_{j}$ of $P$. For $i< j$ we let $\Theta_{i,j}(z)$ denote the $(i,j)$-th block of $\Theta$, which is a $b_i\times b_j$ matrix and $\Theta_1,\ldots,\Theta_s$ denote its diagonal blocks. We infer from Definition \ref{def:admissible} that the pair $(P,\Theta)$ is admissible if and only if $\puip(P)\Theta=AP$, that is, if and only if 
	\begin{equation}
	\label{eqQj}
	\forall j \in \{1,\ldots,s\},\, Q_j(z)  = A^{-1}(z)Q_j(z^p)\Theta_j  + \sum\limits_{i=1}^{j-1} A^{-1}(z)Q_i(z^p)\Theta_{i,j}(z)\, .
	\end{equation}
	By assumption, we may write  
	$$
	Q_j(z)=: \sum\limits_{k\geq \nu_P} Q_{j,k} z^k,\quad \Theta_{i,j}(z)=:\sum\limits_{k= \nu_\Theta}^0 \Theta_{i,j,k} z^k, \quad A^{-1}(z) =: \sum\limits_{k\geq \vb} B_k z^k\, ,
	$$
	where $\Theta_{i,j,k}:=0$ when $k\notin \mathcal S_p$.
	Then, \eqref{eqQj} is equivalent to: $\forall j \leq s,\, \forall n\geq \nu$,
	\begin{equation}
	\label{eq:formule_blocks}
	Q_{j,n} 
	=\sum\limits_{k=\vb}^{n-p\nu_P}B_kQ_{j,\frac{n-k}{p}}\Theta_j + \sum\limits_{i=1}^{j-1}\sum\limits_{k=\vb}^{n-p\nu_P-\nu_\Theta}\sum\limits_{l=\nu_\Theta}^{0}B_kQ_{i,\frac{n-k-l}{p}}\Theta_{i,j,l}\, .
	\end{equation}
	
	\begin{rem}
		\label{rem:mu}
		Since $-\vb-\nu_\Theta\leq(p-1)\mu$, if $n$ is an integer such that $\nu\leq n\leq \mu$ (respectively $n>\mu$) then $\frac{n-\vb-\nu_\Theta}{p}\leq \mu$ (respectively $<n$).  Thus, for any $n \in \{\nu,\ldots,\mu\}$ the matrices $Q_{i,t}$ and $Q_{j,t}$ appearing in the right-hand side of \eqref{eq:formule_blocks} are such that $t\leq \mu$. Furthermore, if $n >\mu$, the matrices $Q_{i,t}$ and $Q_{j,t}$ appearing in the right-hand side of \eqref{eq:formule_blocks} are such that $t < n$. Thus, once the matrices $Q_{j,n}$, $n \in \{\nu,\ldots,\mu\}$, $\Theta_j$ and $\Theta_{i,j,k}$ are known, the matrices $Q_{j,n}$ with $n>\mu$ can be computed inductively. 
	\end{rem}

	\subsubsection{The window $\{\nu,\ldots,\mu\}$}\label{subsec:matM}
	We consider the map $\pi :\K((z))^m \to \K^{m(\mu-\nu+1)}$ defined by
	$$
	\pi : \f(z)= \sum_{n}\f_n z^n \in \K((z))^m\mapsto \prescript{\rm t}{}(\prescript{\rm t}{}\f_\nu,\ldots\prescript{\rm t}{}\f_\mu)\,.
	$$
	Concretely, $\pi$ maps any vector of Laurent series with size $m$ to the vector obtained by concatenating the coefficients corresponding to $z^\nu,z^{\nu+1},\ldots,z^\mu$. Let $F$ be a matrix whose columns $F_i$ are vectors of Laurent series of size $m$. By abuse of notation we denote  $\pi(F)$ the matrix whose columns are the $\pi(F_i)$.  We let $V_0$ denote the vector space $\{0\}^{m(\nu_P-\nu)}\times \K^{m(\mu-\nu_P+1)}$, that is the image under $\pi$ of the space of vectors $\f\in \K((z))^m$ for which $\val \f \geq \nu_P$. In particular, if $(P,\Theta)$ is an admissible pair as in Proposition \ref{prop:incomplete_basis}, and $P_1,\ldots,P_m$ are the columns of $P$, we have $\pi(P_i) \in V_0$.
	For each non-positive integer $l \in \{\nu_\Theta,\ldots,0\}\cap \mathcal S_p=:\mathcal S'_{p}$ we let $M_l$ denote the square matrix defined by 
	\begin{equation}\label{eq:def_Mi}
	M_l\pi(\f(z))=\pi(z^{l}A^{-1}(z)\f(z^p))
	\end{equation}
	for all $\f\in  \K((z))^m$ for which $\val \f \geq \nu_P$, and which is null on the supplementary space $\K^{m(\nu_P-\nu)}\times \{0\}^{m(\mu-\nu_P+1)}$ of $V_0$. From Remark~\ref{rem:mu}, such matrices are well-defined.  Note, for computational considerations, that these matrices $M_l$ admit decompositions into blocks of $m\times m$ matrices denoted by $M_{l;i,j}$, $1\leq i,j\leq \mu-\nu+1$, such that: 
	$$
	M_{l;i,j} = \left\{\begin{array}{ll}
	0 & \text{if } j\in\{ 1,\ldots, \nu_P -\nu \} \\
	B_{i+\nu-l-1 -p(j+\nu-1)}  & \text{otherwise.} 
	\end{array}\right.
	$$

	\subsubsection{Reduction to a finite dimensional problem}

    From now on, we let $M:=M_0$. 
	The map $\pi$ and the matrices $M=M_0,M_{-1},\ldots,M_{\nu_\Theta}$ allow us to express the admissibility of a pair $(P,\Theta)$. Indeed, the equation \eqref{eq:formule_blocks} for $n \in \{\nu,\ldots,\mu\}$ is equivalent to   \begin{equation}\label{eq:formule_blocks2}
	\forall j\in\{1,\ldots,r\}, \quad 
	\pi(Q_j)=M\pi(Q_j)\Theta_j + \sum_{i<j}\sum_{l \in\mathcal S'_{p}}M_l\pi(Q_i)\Theta_{i,j,l}\,.
	\end{equation}
	Conversely, we have the following.
	\begin{prop}\label{prop:window_blocks}
		Let $r\geq 1$ and let $b_1,\ldots,b_r$ be integers with $b_1+\cdots+b_r=m$. For each $j \in \{1,\ldots,r\}$, let $E_j$ be a $m(\mu-\nu+1)\times b_j$ matrix. Suppose that the $m$ columns of $E_1,\ldots,E_r$ are $\K$-linearly independent and belong to $V_0$. Let $\Theta_1,\ldots,\Theta_r$, $\Theta_{i,j,l}$, $l \in \mathcal S'_p$, be matrices with entries in $\K$, such that
		\begin{equation}\label{eq:formule_blocks3}
		\forall j\in\{1,\ldots,r\}, \quad E_j=ME_j\Theta_j + \sum_{i<j}\sum_{l \in\mathcal S'_{p}} M_lE_j\Theta_{i,j,l}\,.
		\end{equation}
		Let $\Theta$ denote the $m\times m$ matrix with diagonal blocks $\Theta_1,\ldots,\Theta_r$ and upper diagonal blocks $\Theta_{i,j}(z)=\sum_{l \in \mathcal S'_p}\Theta_{i,j,l}z^l$. Then, there exists a matrix $P\in\GL_m(\K((z)))$ such that the pair $(P,\Theta)$ is admissible for \eqref{eq:Mahler_system} and $\pi(P)=(E_1\vert \cdots \vert E_r)$.
	\end{prop}
	
	We establish Proposition \ref{prop:window_blocks} after the following lemma.
	\begin{lem}
		\label{lem:indep2}
		Let $P$ be an $m\times m$ matrix with entries in $\K((z))$. Let $\Theta \in \GL_r(\K[z^{-1}])$ be block upper triangular with constant diagonal blocks and support in $\mathcal S_p$. Suppose that $\puip(P)\Theta=AP$. If the columns of $P$ are linearly independent over $\K$, then they are linearly independent over $\K((z))$, that is $P \in \GL_m(\K((z)))$, or, equivalently, $(P,\Theta)$ is an admissible pair.
	\end{lem}
	\begin{proof}
		Let $P_1,\ldots,P_m$ denote the columns of $P$. Since the entries of $P$ are Laurent series with coefficients in $\K$, and $P_1,\ldots,P_m$ are $\K$-linearly independent, they are $\KK$-linearly independent. Thus, up to enlarge $\K$, we may suppose that $\Theta$ is upper triangular with constant diagonal entries $\theta_1,\ldots,\theta_m$. We argue by contradiction. Assume on the contrary that $P_1,\ldots,P_m$ are $\K((z))$-linearly \textit{dependent}. Let $r$ be minimal such that $P_1,\ldots,P_r$ are $\K((z))$-linearly dependent. Let $\lambd(z)=\prescript{\rm t}{}{(\lambda_1(z),\ldots,\lambda_{r}(z),0,\ldots,0)} \in \K((z))^{m}\setminus \{0\}$ be such that
		$$
		\lambda_1P_1+\cdots+\lambda_{r}P_{r}=0\,.
		$$
		By minimality of $r$, $\lambda_{r}\neq 0$. We may then suppose that $\lambda_{r}=1$. Since $P_1,\ldots,P_{r}$ are $\K$-linearly independent, there exists an index $i\in\{1,\ldots,r-1\}$ such that $\lambda_{i} \notin \K$. We let $i_0$ be the greatest of such indices. From the identity
		$$
		\puip(P) \Theta \lambd = AP\lambd=0,
		$$
		the fact that $\Theta$ is upper triangular and the fact that $r$ is minimal, we deduce that there exists $h(z) \in \K((z))$ such that
		$$
		\Theta(z)\lambd(z)=h(z)\lambd(z^p).
		$$
		Looking at the $r$th coordinate on both side, we get $h(z)=\theta_r$. Thus, 
		$$
		\Theta(z)\lambd(z)=\theta_r\lambd(z^p).
		$$
		Let $(0,\ldots,0,\theta_{i_0},t_{i_0+1}(z),\ldots,t_{m}(z))$, $t_i \in \K[z^{-1}]$, denote the $i_0$th row of $\Theta$. Then, 
		$$
		(0,\ldots,0,\theta_{i_0},t_{i_0+1}(z),\ldots,t_{m}(z))\lambd(z)=\theta_r\lambda_{i_0}(z^p)
		$$
		thus
		\begin{equation}\label{eq:lambda_i0}
		\theta_{i_0}\lambda_{i_0}(z)+t(z)=\theta_r\lambda_{i_0}(z^p)
		\end{equation}
		where $t(z) \in \K[z^{-1}]$ has support in $\mathcal S_p$. Looking at the valuation on both sides, one of the following mutually exclusive situations holds:
		\begin{itemize}
			\item[(a)] $\val \lambda_{i_0}=\val t < p \val \lambda_{i_0}$,
			\item[(b)] $\val \lambda_{i_0}=p \val \lambda_{i_0}< \val t$,
			\item[(c)] $ \val t = p\val \lambda_{i_0} $.
		\end{itemize}
		The first two situations are impossible since $\val t \leq 0$. Thus (c) holds and $\val \lambda_{i_0} = \frac{\val t}{p}$. Since $\val \lambda_{i_0}$ is an integer and $\val t \in \mathcal S_p$, we have $\val t =0$ that is $t \in \K$ and $\val \lambda_{i_0}=0$. Then, Equation \eqref{eq:lambda_i0} forces $\lambda_{i_0}$ to be in $\K$, a contradiction.
	\end{proof}

	\begin{proof}[Proof of Proposition \ref{prop:window_blocks}]
		For $j \in\{1,\ldots,r\}$, we define $m\times b_j$ matrices of Laurent series $Q_j=\sum_{n\geq \nu_P} Q_{j,n}z^n$ in the following way:
		\begin{itemize}
			\item $Q_{j,\nu_P},\ldots,Q_{j,\mu}$ are uniquely defined by the identity $\pi(Q_j)=E_j$.
			\item Suppose that $n\geq \mu$ and that $Q_{i,k}$ is defined for any $i\leq j$ and $k<n$. Then $Q_{j,n}$ is uniquely defined by \eqref{eq:formule_blocks}. 
		\end{itemize}
		By \eqref{eq:formule_blocks2} and \eqref{eq:formule_blocks3}, $Q_1,\ldots,Q_r,\Theta$ satisfy \eqref{eq:formule_blocks}. Let $P=(Q_1\vert \cdots \vert Q_r)$, so that $\puip(P)\Theta=AP$. By construction, $\pi(P)=(E_1\vert \cdots \vert E_r)$. Thus, it only remains to prove that the columns of $P$ are $\K((z))$-linearly independent. Since the columns of $E_1,\ldots,E_r$ are $\K$-linearly independent, the columns of $P$ are $\K$-linearly independent. $P_1,\ldots,P_r$. Then, the result follows from Lemma \ref{lem:indep2}.
	\end{proof}
	
	\subsection{Independence of the columns' projections}
	We continue with $\K$ a field that we do not assume to be algebraically closed. 
	Proposition \ref{prop:window_blocks} implies that we only need to compute $\pi(P)$ and $\Theta$ in order to obtain any coefficient of the matrix $P$ of an admissible pair $(P,\Theta)$. The following lemma guarantees that the matrix $\pi(P)$ has maximal rank.
	
	\begin{lem}\label{lem:indep}
		Let $(P,\Theta)$ be an admissible pair for \eqref{eq:Mahler_system_d} with $P \in \GL_m(\K((z))$ and $\Theta$ whose support is in $\mathcal S_p$. Then, the $m$ columns of $\pi(P)$ are $\K$-linearly independent.
	\end{lem}
	\begin{proof}
		Let $P_1,\ldots,P_m$ denote the columns of $P$. Suppose on the contrary that the column vectors $\pi(P_1),\ldots,\pi(P_m)$ are $\K$-linearly dependent and let $\lambda_1,\ldots,\lambda_m \in \K$, not all zero, be such that
		$$
		\lambda_1\pi(P_1)+\cdots+\lambda_m\pi(P_m)=0\,.
		$$
		Let $\lambd=(\lambda_1,\ldots,\lambda_m)$. Then, by definition of $\pi$ and $\mu$, 
		$$
		\val(P\lambd)>\mu \geq \frac{\val \det(A)}{p-1}-(m-1)\nu_P.
		$$
		Let $i_0$ be such that $\lambda_{i_0}\neq 0$ and let $\Lambda$ denote the matrix which is the identity but for the $i_0$th column which is equal to $\lambd$. 
		Then, the matrix $P\Lambda$ has $m-1$ columns with valuation at least $\nu_P$ and one column with valuation at least $\frac{\val \det(A)}{p-1}-(m-1)\nu_P$. Thus,
		$$
		\val \det(P\Lambda) > (m-1)\nu_P +\frac{\val \det(A)}{p-1}-(m-1)\nu_P = \frac{\val \det(A)}{p-1}\,.
		$$
		Meanwhile, $\val \det(P\Lambda)=\val \det(P)$, which is equal to $\frac{\val \det(A)}{p-1}$ from \eqref{eq:val_detP}. This is a contradiction.
	\end{proof}

	\subsection{A naive approach that we will not pursue}\label{sec:algo_colonnes}
	
	Let $\K$ denote an effective field. A naive approach for computing $(P,\Theta)$ would consist in computing the columns of $P$ inductively. Indeed, the algebraic closure $\KK$ of $\K$ remains effective, so that we could assume $\K$ to be algebraically closed. Then, it is sufficient to find a pair for which $\Theta$ is upper triangular. Then, \eqref{eq:formule_blocks3} with $r=m$ and $b_1=\cdots=b_m=1$ provides an equation for the $j$th column of $\pi(P)$ if one knows the previous columns. Precisely, we must look for a pair $(\v,\lambda) \in V_0 \times \K^\times$ such that
	\begin{equation}\label{eq:columns}
	({\rm I} - \lambda  M) \v \in \mathfrak E_{j}
	\end{equation}
	where $\mathfrak E_{j}$ is the vector space spanned by the vectors $M_lE_i$, $l \in \mathcal S'_p$, $1\leq i \leq j-1$, and $E_1,\ldots,E_{j-1}$ are the first columns of $\pi(P)$ that we already computed. Furthermore, we know that $\lambda$ is one of the diagonal entries of $\Theta$, which means that it is an exponent of the system, following the terminology introduced in \cite{Ro18}. Since the set of exponents of \eqref{eq:Mahler_system_d} is computable and finite, solving Equation \eqref{eq:columns} reduces to an elementary problem of linear algebra. One may compute this way $m$ columns $E_1,\ldots,E_m$ satisfying \eqref{eq:formule_blocks3} and obtain this way an admissible pair $(P,\Theta)$ as in Proposition \ref{prop:window_blocks}. 
	
	However, this method has several drawbacks compared to the algorithm we present in the next section. In particular, it unnecessarily requires working over the algebraic closure of $\K$.

	\section{Computing the Puiseux part of a fundamental matrix of solutions}\label{sec:algo_blocs}
	
	The approach described in Section \ref{sec:algo_colonnes} requires working over an algebraic closure of $\K$. This can be avoided as we shall now see. We let $\K$ be an effective field that we do not assume to be algebraically closed. We consider a Mahler system
	\begin{equation}\label{eq:Mahler_system}
	\puip(Y)=AY,\quad A(z) \in \GL_m(\K((z)))\, 
	\end{equation}
	for which $d(A)=1$. 
	Proposition \ref{prop:ramification} does not allow us to conclude that there exists a pair $(P,\Theta)$ admissible with respect to \eqref{eq:Mahler_system}, with $P \in \GL_m(\K((z)))$, since $\K$ is not necessarily algebraically closed. This fact will follow from our construction. Recall that if there exists an admissible pair $(P,\Theta)$ and $\Theta$ has its support in $\mathcal S_p$, then the vector space spanned by the columns of $\pi(P)$ has dimension $m$ and is included in 
	$$
	V_0 := \{ 0 \}^{m(\nu_P - \nu)}\times \K^{m(\mu-\nu_P+1)}
	$$
	since $\val(P) \geq \nu_P$. According to the block structure of $\Theta$, one can decompose $P$ as $(Q_1\vert \cdots \vert Q_r)$ where $Q_i$ are matrices with $m$ rows each. Our construction consists of computing a sequence of vector spaces $\mathfrak{X}_1\subset \mathfrak{X}_2 \subset \cdots \subset V_0$ with the property that $\pi(Q_i) \subset \mathfrak X_i$ for each $i$, for any such admissible pair $(P,\Theta)$. We establish that one of these vector spaces has dimension $m$ at some point. By choosing an appropriate basis, we are able to build an admissible pair.

	\subsection{Construction of the vector spaces} Recall that $\mathcal S'_p := \mathcal S_p \cap \{\nu_\Theta,\ldots,0\}$.   We define a sequence of vector spaces $(\mathfrak X_j)_{j\geq 0}$ by induction on $j$ as follows. 
	\begin{itemize}
		\item We let $\mathfrak X_0 := \{ 0\}$.
		\item Let $j \in \Z_{\geq 0}$ and assume that $\X_j$ has been defined. We let $\U_j := {\rm span}_\K(M_k\X_j\,:\,k\in \mathcal S'_p )$ and $\X_{j+1}$ be the largest subspace of $V_0$ for which
		\begin{equation}
		\label{eq:Xj}
		M\X_{j+1}\subset \X_{j+1}+\U_j\quad \text{ and } \quad \X_{j+1}\subset M\X_{j+1}+\U_j.
		\end{equation}
	\end{itemize}
	Notice that $\mathfrak U_0 = \{ 0\}$ and $\X_1$ is the largest subspace of $V_0$ such that $M\X_1=\X_1$.
	\begin{lem}
		\label{lem:inclusion}
		The sequences $(\X_j)_{j\geq 0}$ and $(\U_j)_{j\geq 0}$ are non-decreasing and, for all $j \in \Z_{\geq 0}$, $\X_j\subset \U_j$. 
	\end{lem}
	\begin{proof}
		We prove by induction on $j\geq 0$ that $\X_j\subset \U_j$, $\X_j\subset \X_{j+1}$ and $\U_j\subset \U_{j+1}$. When $j=0$ the result is clear. Let $j\geq 1$ and assume the property true with $j-1$. 
		By definition, $\X_j\subset M\X_j + \U_{j-1}$. 
		Since $M = M_0 $ and $0\in \mathcal{S}'_p$, we have $M\X_j\subset \U_j$. Moreover, by induction hypothesis, $\U_{j-1}\subset \U_j$. Thus, $\X_j\subset \U_j$. Since $\X_j\subset V_0$, to prove that $\X_j\subset \X_{j+1}$ we just have to prove that $\X_j$ satisfies \eqref{eq:Xj}. This follows from the fact that $M\X_j\subset \U_j$ and  $\X_j\subset \U_j$. The inclusion $\U_j\subset \U_{j+1}$ is an immediate consequence of the fact that $\X_j\subset \X_{j+1}$.
	\end{proof}

	Before going to our main result, we provide a method to compute the spaces $\mathfrak X_j$. Consider the non-increasing sequence $(\mathfrak F_{j,\ell})_{\ell\geq 0}$ of vector spaces defined by $\mathfrak F_{j,0} := V_0$ and, for any integer $\ell\geq 0$, 
	$$
	\mathfrak F _{j,\ell+1} := \mathfrak{F}_{j,\ell}\cap M^{-1} (\mathfrak{F}_{j,\ell} + \U_j) \cap (M\mathfrak{F}_{j,\ell}+\U_j)\, .
	$$
	\begin{lem}
		\label{lem:constr_Xj} Let $j \geq 0$ and let $\ell_j$ be the least integer for which $\mathfrak{F}_{j,\ell_j}=\mathfrak{F}_{j,\ell_j+1}$. Then $\mathfrak X_{j+1} =  \mathfrak F_{j,\ell_j}$ and $\ell_j\leq  m(\mu-\nu_P+1)$.
	\end{lem}
	\begin{proof} First, we prove by induction on $\ell\geq 0$ that $\mathfrak X_{j+1} \subset \mathfrak F_{j,\ell}$. By assumption, $\mathfrak X_{j+1} \subset \mathfrak F_{j,0}= V_0$. Let $\ell \geq 0$ and assume $\mathfrak X_{j+1} \subset \mathfrak F_{j,\ell}$. Then, it follows from \eqref{eq:Xj} that $\mathfrak{X}_{j+1} \subset  \mathfrak{F}_{j,\ell+1}$, which ends the induction. 
		
		The sequence  $(\mathfrak{F}_{j,\ell})_\ell$ is a non-increasing sequence of sub-vector spaces of $V_0$. Since $V_0$ is finite dimensional, there exists a non-negative integer $\ell$ such that $\mathfrak{F}_{j,\ell}=\mathfrak{F}_{j,\ell+1}$. Let $\ell_j$ denote the least of all such integers. Then, 
		$
		\mathfrak{F}_{j,\ell_j} \subset M^{-1} (\mathfrak{F}_{j,\ell_j} + \U_j) \cap (M\mathfrak{F}_{j,\ell_j}+\U_j)
		$ which implies that  $\mathfrak{F}_{j,\ell_j} $ satisfies \eqref{eq:Xj}. Since $\mathfrak{F}_{j,\ell_j}\subset V_0$, by maximality of $\mathfrak X_{j+1}$, we have $\mathfrak X_{j+1}\supset \mathfrak{F}_{j,\ell_j}$. Then, the beginning of the proof implies that $\mathfrak{X}_{j+1} = \mathfrak{F}_{j,\ell_j}$. Since $\dim(\mathfrak F_{j,0}) =\dim V_0= m(\mu-\nu_P+1)$, we have $\ell_j\leq  m(\mu-\nu_P+1)$. 
	\end{proof}

	\begin{rem*} 
		In \cite{FP22}, we were computing fundamental matrices of solutions in the case where the system is regular singular at $0$. This corresponds precisely to the case where the matrix $\Theta$ may be taken to be constant. In other words, this corresponds to the case where the vector space $\X_1$ has dimension $m$. Actually, in \cite{FP22} our vector space, called $\X_d$, was the greatest space to satisfy $M_d\X_d=\X_d$ and $\X_d\subset \ker N_d$ for some explicit linear maps $M_d$ and $N_d$. In the case $d=1$, this vector space is closely related to the vector space $\X_1$ defined here. Precisely, the vector space $\X_d$ from \cite{FP22} is the projection of the vector space $\X_1$ onto the space where only the coefficients between $\nu_P$ and $\lceil-\val A^{-1}/(p-1)\rceil$ are taken into account. The map $M_d$ corresponds to the restriction of the map $M$ introduced here to this vector space, and the matrix $N_d$ corresponds to the projection of this map onto its coefficients between $p\nu_P+\val A^{-1}$ and $\nu_P-1$, whenever $p\nu_P+\val A^{-1}\leq \nu_P-1$.
	\end{rem*}
	
	\subsection{Construction of an admissible pair}
	Algorithm \ref{algo:solutions_blocks} relies on the following key-lemma.
	\begin{lem}\label{lem:X_i}
		We have $\dim \X_r\geq m$ for some $r\geq 0$. For such an $r$, there exists a basis $\e_{1,1},\ldots,\e_{1,m_1},\e_{2,1},\ldots,\e_{r,m_r}$ of $\X_r$ fitting with the nested sequence $\X_1\subset \X_2 \subset \cdots \X_r$, such that the following holds: letting $E_j$, $j\in \{1,\ldots,r\}$ denote the matrix with columns $\e_{j,1},\ldots,\e_{j,m_j}$, there exists $\Theta_j\in \GL_{m_j}(\K)$ such that
		\begin{equation}\label{eq:M_E_j}
		E_j - ME_j\Theta_j \in \U_{j-1}\, .
		\end{equation}
	\end{lem}
	
	\begin{rem*}
		Using \eqref{eq:M_E_j}, we can write for each $j\geq 2$,
		$$
		E_j - ME_j\Theta_j = \sum_{i<j}\sum_{k\in \mathcal S'_p } M_kE_i\Theta_{i,j,k}
		$$ 
		for some matrices $\Theta_{i,j,k}$ with coefficients in $\K$. Then, from Proposition \ref{prop:window_blocks}, there is an admissible pair $(P,\Theta)$ with $\pi(P) = (E_1\vert \cdots \vert E_r)$. In particular, $\dim \mathfrak X_r=m$. 
	\end{rem*}
	\begin{proof}[Proof of Lemma \ref{lem:X_i}]
		We start with the first statement.  Let $\KK$ denote the algebraic closure of $\K$ and $V_0(\KK)$, $\X_j(\KK)$, $\U_j(\KK)$ denote the $\KK$-vector spaces spanned by $V_0$, $\X_j$ and $\U_j$ respectively. Note that $\dim_{\KK}\X_j(\KK)=\dim_\K \X_j$. For any $j \in \{1,\ldots,r\}$, we claim that $\X_j(\KK)$ is the greatest $\KK$-vector space $\X$ for which $M\X\subset \X+\U_{j-1}(\KK)$ and $\X \subset M\X + \U_{j-1}(\KK)$. Indeed, since $M$ has coefficients in $\K$, such equations are defined over $\K$ and the result follows\footnote{This also follows by applying Lemma~\ref{lem:constr_Xj} twice, once with $\K$ and once with $\KK$ as a base field, and noticing that the vector spaces $\mathfrak F_{j,\ell}$ built in the latter case are obtained from the ones built in the former case by extension of the scalars from $\K$ to $\KK$.} from the definition of $\X_j$. 
		
		Let $(P,\Theta)$ be admissible for the system over $\KK((z))$, with $P\in\GL_m(\KK((z))$. Such a pair exists thanks to Proposition \ref{prop:ramification}. Let $P=(Q_1\vert \cdots \vert Q_r)$ denote the decomposition of $P$ fitting with the block decomposition of $\Theta$. Let $E_j:=\pi(Q_j)$ and let $\mathfrak E_j$ be the $\KK$-vector space spanned by its columns. By Lemma \ref{lem:indep}, the dimension of $\mathfrak E_1+\cdots+ \mathfrak E_r$ is equal to $m$.  We are going to prove by induction on $j\geq 1$ that $\mathfrak E_j \subset \X_j(\KK)$. When $j = 1$, \eqref{eq:formule_blocks} implies that $E_1=ME_1\Theta_1$, which implies that $M\mathfrak E_1=\mathfrak E_1$. By maximality of $\X_1(\KK)$, $\mathfrak E_1\subset \mathfrak X_1(\KK)$. Let $j\geq 2$ and suppose that $\mathfrak E_{i} \subset \mathfrak X_i(\KK)$ for every $i <j$. From \eqref{eq:formule_blocks} we deduce that
		$$
		E_j=ME_j\Theta_j + \sum_{i<j}\sum_{l=\nu_\Theta}^0 M_lE_i\Theta_{i,j,l}\,.
		$$
		Since $\mathfrak E_{i} \subset \mathfrak X_i(\KK)$ for every $i <j$, the column vectors of the second term in the right-hand side belong to $\mathfrak U_{j-1}(\KK)$. Since $\Theta_j$ is invertible, we get that
		$$
		M\mathfrak E_j \subset \mathfrak E_j + \mathfrak U_{j-1}(\KK) \text{ and } \mathfrak E_j \subset M\mathfrak E_j + \mathfrak U_{j-1}(\KK)\,.
		$$
		By maximality of $\X_j(\KK)$, $\mathfrak E_j \subset \mathfrak X_j(\KK)$. This ends the induction. Now, since $\X_i\subset \X_j$ when $i\leq j$, we have $\mathfrak E_1+\cdots+ \mathfrak E_r\subset \X_r(\KK)$. Thus, 
		$$\dim_\K \X_r=\dim_{\KK} \X_r(\KK) \geq \dim_\KK \mathfrak E_1+\cdots+ \mathfrak E_r=m.
		$$

		\medskip
		Let us now prove that we can find a basis of $\mathfrak X_r$ as mentioned in the lemma. We only have to prove by induction on $j$ that, for each $j\geq 1$, we can find $\e_{j,1},\ldots,\e_{j,m_j}$ with the desired properties. 
		
		Let $j=1$. Then, any basis $\e_{1,1},\ldots,\e_{1,m_1}$ of $\X_1$ has the required property, since in that case \eqref{eq:M_E_j} reads $E_1=ME_1\Theta_1$, which is true since $\X_1=M\X_1$. 
		
		Suppose $j\geq 2$ and that we have built $\e_{1,1},\ldots,\e_{j-1,m_{j-1}}$. Since $\X_{j-1}\subset\U_{j-1} \cap \X_j$ by Lemma \ref{lem:inclusion}, we may decompose $\X_j$ as follows
		$$
		\X_j=\X_{j-1} \oplus \mathfrak Y \oplus \mathfrak Z
		$$
		where $\mathfrak Y$ is any supplementary space of $\X_{j-1}$ within $\U_{j-1} \cap \X_j$ and $\mathfrak Z$ is any supplementary space of $\X_{j-1} \oplus \mathfrak Y(=\mathfrak X_j\cap\mathfrak U_{j-1})$ within $\X_j$. Consider a basis $\e_{j,1},\ldots,\e_{j,s}$ of $\mathfrak Y$ and $\e_{j,s+1},\ldots,\e_{j,m_j}$ of $\mathfrak Z$. Let $E_j$ be the associated matrix and write $E_j=:(E_{\mathfrak Y}\vert E_{\mathfrak Z})$ its decomposition according to this cut\footnote{If $\mathfrak Y$ (resp. $\mathfrak Z$) is the zero space then $E_{\mathfrak Y}$ (resp. $E_{\mathfrak Z}$) is the empty matrix.}. Since $M\mathfrak X_j\subset \mathfrak X_j + \mathfrak U_{j-1} = \mathfrak Z + \U_{j-1}$,  we have $M(\mathfrak Y + \mathfrak Z)  \subset \mathfrak Z + \U_{j-1}$. Thus, we may write
		\begin{equation}\label{eq:introduce_R}
		ME_j - E_{\mathfrak Z}R \in \U_{j-1} 
		\end{equation}
		for some rectangular matrix $R$. Let us prove that $R$ has maximal rank. Let $\mathfrak V \subset \mathfrak Z$ be the vector space spanned by the columns of $E_{\mathfrak Z}R$. Then, \eqref{eq:introduce_R} implies that $M\mathfrak (\mathfrak Y + \mathfrak Z)\subset \mathfrak V+ \U_{j-1}$. Let $\v \in \mathfrak Z$. Then, 
		\begin{multline*}\v\in \mathfrak Z\subset \X_j \subset M\X_j+\U_{j-1} \subset M\X_{j-1} + M(\mathfrak Y + \mathfrak Z) +\U_{j-1} 
		\\ \subset \X_{j-1}+\U_{j-2}+ \mathfrak V+\U_{j-1} \subset \mathfrak V+\U_{j-1}
		\end{multline*} 
		by Lemma \ref{lem:inclusion}.
		Thus, there exists $\w \in \mathfrak V$ such that $\v+\w \in \U_{j-1}$. Meanwhile $\v + \w \in \mathfrak Z$, since $\mathfrak V \subset \mathfrak Z$. Since $\mathfrak Z \cap \U_{j-1}= \{0\}$, $\v+\w=0$, that is $\v=-\w \in \mathfrak V$. Hence $\mathfrak V= \mathfrak Z$, that is, $R$ has maximal rank.
		Thus, one can complete $R$ into an invertible matrix $V$. Then, since the columns of $E_{\mathfrak Y}$ belongs to $\U_{j-1}$, we have 
		$$
		ME_j-E_jV \in \U_{j-1}
		$$
		We conclude by setting $\Theta_j:=V^{-1}$. 
	\end{proof}
	
	\medskip
	
	\begin{algorithm}[H]\label{algo:solutions_blocks}
		\SetAlgoLined
		\KwIn{An integer $p\geq 2$ and a square matrix $A$ with coefficients in $\K((z))$ such that $d(A)=1$.}
		\KwOut{A matrix $\Theta$ and a truncation of a matrix $P$ up to order $\mu$ such that $(P,\Theta)$ is an admissible pair with respect to the system \eqref{eq:Mahler_sys}.}
		Set $\mathfrak X:=\{ 0\}$.\tcp*[f]{ \small {$\mathfrak X$ will successively be equal to $\mathfrak X_0,\mathfrak X_1,\X_2,\ldots$}}
		\\ Let $\mathcal E$ and $\mathcal U$ denote the empty lists.
		\tcp*[f]{\small They will respectively host the matrices $E_1,\ldots,E_r$ and the spaces $\mathfrak U_0,\ldots,\mathfrak U_{r-1}.$}
		\\ 	\While{$\dim(\mathfrak X)<m$}{
			\tcp*[f]{\small Inductive computation of $\mathfrak X_1,\ldots,\X_r$ and of $E_1,\ldots,E_r$.}
			\\ Set $\mathfrak U:={\rm span}_\K(M_k\X\,:\,k\in \mathcal S'_p )$.
			\\ Set $\mathfrak F:=V_0$.
			\\ Set $\mathfrak G:=\mathfrak F \cap M^{-1} (\mathfrak F + \mathfrak U)\cap (M\mathfrak F+\mathfrak U)$.
			\\ \While{$\mathfrak F\neq \mathfrak G$}{
				Set $\mathfrak F:= \mathfrak G$.
				\\ Set $\mathfrak G:=\mathfrak F \cap M^{-1} (\mathfrak F + \mathfrak U)\cap (M\mathfrak F+\mathfrak U)$.
			}
			Let $\mathfrak Y$ be a supplementary space of $\mathfrak X$ within $\mathfrak U \cap \mathfrak F$.
			\\ Let $\mathfrak Z$ be any supplementary space of $\mathfrak X \oplus \mathfrak Y$ within $\mathfrak F$.
			\\ Let $E_{\mathfrak Y}$ (\textit{resp.~}$E_{\mathfrak Z}$) be a matrix whose columns form a basis of $\mathfrak Y$ (\textit{resp.~}$\mathfrak Z$) and set $E=(E_{\mathfrak Y}\vert E_{\mathfrak Z})$. 
			\\ Append $E$ to the list $\mathcal E$ and $\mathfrak U$ to the list $\mathcal U$.
			\\ Set $\mathfrak X:= \mathfrak F$.}	\tcp*[f]{\small At this stage $\mathcal E=(E_1,\ldots,E_r)$, $\mathcal U=(\mathfrak U_0,\ldots,\mathfrak U_{r-1})$ and $\mathfrak X=\mathfrak X_r$. }
		\\ Set $r:={\rm card}(\mathcal E)$.
		\\  \For{$j$ \emph{from} $1$ \emph{to} $r$}{Compute an invertible matrix $\Theta_j$ with entries in $\K$ such that $E_j -ME_j\Theta_j \in \mathfrak U_{j-1}$. 
			\\ Let $\Theta_{i,j,k}$ be the constant matrices such that $E_j - ME_j\Theta_j = \sum_{i<j}\sum_{k\in \mathcal S'_p } M_kE_i\Theta_{i,j,k}$.
			\\ Let $P_j(z)= \sum_{k=\nu_P}^\mu Q_{j,k}z^k$ be the matrix of Laurent polynomial of valuation at least $\nu_P$ and degree at most $\mu$ such that $\pi(P_j) =E_j$.}
		Let $\overline{P}(z) := (P_1(z) \mid \cdots\mid P_r(z))$ and let $\Theta$ be the block upper triangular matrix with diagonal blocks $\Theta_1,\ldots, \Theta_r$ and whose $(i,j)$th block is $\Theta_{i,j}(z)=\sum_{k\in\mathcal{S}_p'}\Theta_{i,j,k}z^k$ when $i<j$.
		\\  \Return $(\overline{P},\Theta)$.
		\caption{An algorithm to compute $\Theta$ and a truncation of $P$ at order $\mu$ where $(P,\Theta)$ is an admissible pair.}
	\end{algorithm}
	
	\medskip
	
	Algorithm~\ref{algo:solutions_blocks} returns the truncation of $P$ up to order $\mu$. Indeed, it is not difficult to check that at the end of the first `while' loop we have $\X = \X_r$ and the matrices $E_1,\ldots,E_r$ satisfy \eqref{eq:formule_blocks3}. Thus, the result follows from Proposition \ref{prop:window_blocks}. Then, thanks to Remark \ref{rem:mu}, it is an easy task to compute coefficients of $P$ with higher orders. Note that $r \leq m$, since the sequence $(\mathfrak X_i)_{1\leq i \leq r}$ is increasing and, if we had $\mathfrak X_{i_0}=\mathfrak X_{i_0+1}$ for some $i_0 <r$, the sequence would be stationary from this index. Furthermore, for each $j$,  Lemma \ref{lem:constr_Xj} implies that $\ell_j \leq m(\mu-\nu_P+1)$. Thus the second `while' loop ends after at most $ m(\mu-\nu_P+1)$ iterations. In the end, it is possible to bound the algorithm's complexity. We shall not provide details here.

	\section{Computing the Hahn and the constant part of a fundamental matrix of solutions}\label{sec:HeC}
	Throughout Section~\ref{sec:HeC}, we assume that $\K$ is an effective field which is algebraically closed. This restriction is necessary to obtain the matrix $e_C$ in the required form. The matrix $H$ could be computed over any field, but presenting it in this generality would significantly complicate the arguments; therefore, we restrict to the algebraically closed case. We fix an equation \eqref{eq:Mahler_system} for which $d(A)=1$. We fix an admissible pair $(P,\Theta)$ for which we assume $P \in \GL_m(\K((z)))$ and $\Theta \in \GL_m(\K[z^{-1}])$. Up to a gauge transformation, we may suppose that $\Theta$ is upper triangular. Our goal is to describe the matrices $H$ and $e_C$ that intervene in the decomposition \eqref{eq:mat_fund_sol} of a fundamental matrix of solutions. Recall that they satisfy
    $$
    \puip(H)C=\Theta H\quad \text{and } \quad \puip(e_C)=Ce_C
    $$
    where $C$ is the (upper triangular) constant term of $\Theta$.
	
	\subsection{Computing the matrix of Hahn series}\label{sec:Hahn}
	
	The general form \eqref{eq:mat_fund_sol} of a fundamental matrix of solutions guarantees that there exists a matrix $H \in \GL_m(\Hahn_\K)$ such that $\puip(H)C=\Theta H$, where $C$ is the constant part of $\Theta$. More precisely, it is established in \cite{FR25} that one may take $H$ to be upper triangular with only $1$s on the diagonal and upper triangular entries with support in the set of negative rational numbers ---a fact we will reprove here. 
	Let $\theta_{i,j}(z) \in \K[z^{-1}]$, $i<j$, denote the upper diagonal entries of $\Theta$ and $\theta_1,\ldots,\theta_m \in \K$ denote its diagonal entries. Let $c_{i,j}$, $i<j$, denote the upper diagonal entries of $C$, that is, the constant term of $\theta_{i,j}(z)$. An upper triangular matrix $H=(h_{i,j})_{i,j}$ with the identity on the diagonal and upper diagonal entries with support in $\mathbb Q_{<0}$ satisfies $\puip(H)C=\Theta H$ if and only if, for any $i,j$, 
	\begin{equation}\label{eq:formula_Hahn}
	\theta_j h_{i,j}(z^p) - \theta_i h_{i,j}(z) = \sum_{k=i+1}^{j-1}\theta_{i,k}(z) h_{k,j}(z) + \theta_{i,j}(z)-c_{i,j}-\sum_{l=i+1}^{j-1} c_{l,j}h_{i,l}(z^p)\,.
	\end{equation}
	The right-hand side of the equality has support in the set of negative rational numbers since $\theta_{i,j}(z)-c_{i,j} \in z^{-1}\K[z^{-1}]$. Consider the order $\prec$ on $\{(i,j)\,:\,i<j\}$ introduced in the proof of Proposition \ref{prop:incomplete_basis} and defined by $(k,l)\prec(i,j)$ if and only if either $l=j$ and $k>i$, or $l<j$. Then, the $h_{k,j}$ and $h_{i,l}$ intervening in the right-hand side are such that $(k,j) \prec (i,j)$ and $(i,l)\prec(i,j)$. Thus, the Hahn series $h_{i,j}$ may be computed by induction on the pairs $(i,j)$. 
	
	By linearity, one only needs to consider inhomogeneous equations for which the second member has one term. We consider three types of equations. Precisely, given $\kappa,\eta,\tau \in \K^\times$, $\gamma \in \QQ_{>0}$, $s\geq 1$ and $\omeg=(\u,\a)$ with $\u=(u_{k_1,\ldots,k_s})_{k_1,\ldots,k_s}\subset \K$ a multi-linear recurrence sequence and $\a=(a_1,\ldots,a_s)$ we consider the equations
	\begin{align}
	\label{eq:inhomogene} \kappa h(z^p) - \eta h(z) & = \tau z^{-\gamma}
	\\ \label{eq:negativepower} \kappa h(z^p) - \eta h(z) &= \tau z^{-\gamma}\bxi_{\omeg}(z)
	\\ \label{eq:nullpower} \kappa h(z^p) - \eta h(z) & = \tau \bxi_{\omeg}(z)
	\end{align}
			Recall that the Hahn series $\bxi_\omeg$ are defined by 
	\begin{equation*}
	\boldsymbol{\xi}_{\boldsymbol{\omega}} (z) := \sum_{k_1,\cdots,k_s\geq 1} u_{k_1,\ldots,k_s}z^{-\frac{a_1}{p^{k_1}}-\cdots-\frac{a_s}{p^{k_1+\cdots+k_s}}},\quad \text{with }\boldsymbol{\omega}:=(\u,\a),
	\end{equation*}
	where  $\u:=(u_{k_1,\ldots,k_s})_{k_1,\ldots,k_s}\subset \K$ is a multi-linear recurrence sequence and $\a:=(a_1,\ldots,a_s) \in \QQ_{>0}$ for some integer $s\geq 0$.
	
	Let us recall some facts about multi-linear recurrence sequences. We follow the approach from \cite{Schmidt}. Let $s\geq 1$ be an integer and $\boldsymbol X=(X_1,\ldots,X_s)$ be a $s$-tuple of indeterminates. The ring $ \K[\boldsymbol X^{\pm 1}]=\K[X_1^{\pm 1},\ldots,X_s^{\pm 1}]$ acts on the $\K$-vector space of sequences $\K^{\Z^s}$ by
	$$
	(P\u)_{k_1,\ldots,k_s}:= \sum_{i_1,\ldots,i_s\in \Z} p_{i_1,\ldots,i_s} u_{k_1+i_1,\ldots,k_s+i_s}
	$$
	when $P(\boldsymbol X):=\sum_{i_1,\ldots,i_s} p_{i_1,\ldots,i_s}X_1^{i_1}\cdots X_s^{i_s}$ and $\u:=(u_{k_1,\ldots,k_s})_{k_1,\ldots,k_s}$. A sequence $\u \in \K^{\Z^s}$ is multi-linear recurrent if, letting $\I(\u)$ be the ideal $\{P \in \K[\boldsymbol X^{\pm 1}]\,:\, P\u=0\}$, the quotient $\K[\boldsymbol X^{\pm 1}]/\I(\u)$ is a finite dimensional $\K$-vector space. It is actually equivalent requiring that $\u$ is linear recurrent with respect to each of its variables.
	\begin{lem}\label{lem:lin_rec}
		A sequence $\u=(u_{k_1,\ldots,k_s})_{k_1,\ldots,k_s} \in \K^{\Z^s}$ is multi-linear recurrent if and only if, for each $i\in\{1,\ldots,s\}$ there exist an integer $m_i\geq 1$ and coefficients $b_{i,0},\ldots,b_{i,m_i}$, with $b_{i,m_i}b_{i,0}\neq 0$ such that, for every integers $k_1,\ldots,k_s \in \Z$,
		\begin{equation}\label{eq:lin_rec}
		\sum_{j=0}^{m_i} b_{i,j} u_{k_1,\ldots,k_{i-1},k_i+j,k_{i+1},\ldots,k_s}=0\,.
		\end{equation}
	\end{lem}
	\begin{proof}
	We first prove the reverse direction.
	 Identity \eqref{eq:lin_rec} is equivalent to the fact that the polynomials $\sum_{j=0}^{m_i} b_{i,j}X_i^j$ belong to $\I(\u)$. Thus, the monomials $X_1^{j_1}\cdots X_s^{j_s}$, $j_i\in\{0,\ldots,m_i-1\}$, $i \in \{1,\ldots,s\}$, span the vector space $\K[\boldsymbol X^{\pm 1}]/\I(\u)$. \textit{A fortiori}, it is finite dimensional. Thus $\u$ is multi-linear recurrent.
	
	Conversely, suppose that $\u$ is multi-linear recurrent and let $m$ be the dimension of $\K[\boldsymbol X^{\pm 1}]/\I(\u)$. Then, for each $i$, $1,X_i,X_i^2,\ldots,X_i^{m}\in\K[\boldsymbol X^{\pm 1}]$ are linearly dependent modulo $\I(\u)$. Thus, \eqref{eq:lin_rec} holds for some $m_i \leq m$. 
	\end{proof}
	
	Given a multi-linear recurrence sequence $\u$ and a element $\theta \in \K$, the sequence with general term
	$$
    u^{[\theta]}_{l,k_2,\ldots,k_s}:=  \sum_{k_1=1}^{l-1} u_{k_1,\ldots,k_s} \theta^{l-k_1}
	$$
	is a multi-linear recurrence sequence. We denote it by $\u^{[\theta]}$. Indeed, let \eqref{eq:lin_rec} be the equations satisfied by $\u$. The sequence $\u^{[\theta]}$ also satisfies these equations when $i \in\{2,\ldots,s\}$. When $i=1$, since 
	$$
	u^{[\theta]}_{l+1,k_2,\ldots,k_s} -\theta u^{[\theta]}_{l,k_2,\ldots,k_s} = u_{l,\ldots,k_s}
	$$
	we have 
	$$
	b_{1,m}	u^{[\theta]}_{l+m+1,k_2,\ldots,k_s} + \sum_{j=1}^m (b_{1,j-1} - \theta b_{1,j})	u^{[\theta]}_{l+1,k_2,\ldots,k_s} - \theta u^{[\theta]}_{l,k_2,\ldots,k_s} =0.
	$$
	Using Lemma \ref{lem:lin_rec} we conclude that $\u^{[\theta]}$ is a multi-linear recurrence sequence.
	
	\medskip
	We are now ready to solve Equations \eqref{eq:inhomogene}, \eqref{eq:negativepower} and \eqref{eq:nullpower}.
	\begin{lem}\label{lem:Hahn}
		Let $\kappa,\eta,\tau \in \K^\times$ and $\gamma \in \QQ_{>0}$. Let $s\geq 1$ and $\omeg:=(\u,\a)$ with $\u:=(u_{k_1,\ldots,k_s})_{k_1,\ldots,k_s}$ and $\a:=(a_1,\ldots,a_s)$. 
		Then:
		\begin{enumerate}[label=\rm (\roman*)]
			\item \label{item:inhomo} $h(z) :=\frac{\tau}{\eta} \bxi_{((u_k)_k,(\gamma))}(z)$ is a solution to \eqref{eq:inhomogene}, where $u_k :=(\eta\kappa^{-1})^k$;
			\item \label{item:negative} $h(z) :=\frac{\tau}{\eta} \bxi_{(\v,(\gamma,a_1,\ldots,a_s))}(z)$ is a solution to \eqref{eq:negativepower} where $\v:=(v_{k_0,\ldots,k_s})_{k_0,\ldots,k_s}$ with
			$$
			v_{k_0,\ldots,k_s}:=(\eta\kappa^{-1})^{k_0}u_{k_1,\ldots,k_s}
			$$
			\item \label{item:nullpower1} $h(z):= \frac{\tau}{\eta} \bxi_{(\u^{[\eta\kappa^{-1}]},\a)}(z)$
			is a solution to \eqref{eq:nullpower}. 
		\end{enumerate}
	\end{lem}
	\begin{proof}The proof of \ref{item:inhomo} is straightforward. Proving that the series from \ref{item:negative} satisfies \eqref{eq:negativepower} is straighforward too. The fact that the sequence $\v$ is multi-linear recurrent is an immediate consequence of Lemma \ref{lem:lin_rec}. Let us prove \ref{item:nullpower1}. It is not difficult to check that the series
		$$
		h = \frac{\tau}{\eta} \sum_{k\geq 1} \left(\frac{\eta}{\kappa}\right)^{k}\puip^{-k}(\bxi_\omeg)
		$$
		satisfies 
		$$
		\kappa h(z^p) - \eta h(z) = \tau \bxi_{\omeg}(z).
		$$
		The fact that $h(z)$ is a well-defined Hahn series is a consequence of \cite[Lemma~33]{Ro24}. It only remains to prove that $h(z)$ has the desired form. Making a change of variable $l \leftarrow k+k_1$ and setting $\theta := \eta\kappa^{-1}$, we have
		\begin{align*}
		h(z) &= \frac{\tau}{\eta}  \sum_{k\geq 1}\theta^{k}\puip^{-k}(\bxi_\omeg)
		\\ & =\frac{\tau}{\eta}  \sum_{k,k_1,\ldots,k_s\geq 1} u_{k_1,\ldots,k_s} \theta^{k}z^{-\frac{a_1}{p^{k+k_1}}-\cdots - \frac{a_s}{p^{k+k_1+\cdots+k_s}}}
		\\ & =\frac{\tau}{\eta} \sum_{l\geq 2} \left(\sum_{k_1=1}^{l-1}  u_{k_1,\ldots,k_s} \theta^{l-k_1}\right) z^{-\frac{a_1}{p^{l}}-\frac{a_2}{p^{l+k_2}}-\cdots - \frac{a_s}{p^{l+k_2+\cdots+k_s}}}
		\\ & = \frac{\tau}{\eta}  \sum_{l,k_2,\ldots,k_s\geq 1}u^{[\theta]}_{l,k_2,\ldots,k_s}z^{-\frac{a_1}{p^{l}}-\frac{a_2}{p^{l+k_2}}-\cdots - \frac{a_s}{p^{l+k_2+\cdots+k_s}}}
		\\ & = \frac{\tau}{\eta} \bxi_{(\u^{[\theta]},\a)}(z) \, .
		\end{align*}
	\end{proof}

	This enables us to compute the matrix $H$ explicitly.
	
	\medskip
	
	\begin{algorithm}[H]\label{algo:matrixH}
		\SetAlgoLined
		\KwIn{An upper triangular matrix $\Theta$, with non-zero constant diagonal entries and upper diagonal entries in $\K[z^{-1}]$}
		\KwOut{A matrix $H$ with coefficients in ${\rm span}_{\K}\{\bxi_\omeg\,:\,\omeg \in \Omega\}$ such that $\puip(H)C=\Theta H$}
		Set $h_{i,j}:=0$ for $1\leq j<i\leq m$ and $h_{i,i}:=1$ for every $i \in \{1,\ldots,m\}$\;
		\For{$(i,j) \in \{1\leq i<j\leq m\}$ ordered with respect to $\prec$}{
			Using the expressions of $h_{k,l}$, $(k,l)\prec (i,j)$ previously computed, decompose \eqref{eq:formula_Hahn} as a finite sum of equations of the form \eqref{eq:inhomogene}, \eqref{eq:negativepower} and \eqref{eq:nullpower}.
			\\ Solve each of these equations by the formulas given in Lemma \ref{lem:Hahn}.
			\\ Set $h_{i,j}$ to be the sum of these solutions. 
		}
		\Return $H:=(h_{i,j})_{i,j}$.
		\caption{An algorithm to compute the matrix $H$}
	\end{algorithm}
	
		\begin{rem*}
		When $\K$ has characteristic $0$, the sequences with general term
	\begin{equation}\label{eq:form_lin_rec2}
	k_1^{\alpha_1}\cdots k_s^{\alpha_s}\theta_1^{k_1}\cdots \theta_s^{k_s}, \quad \alpha_1,\ldots,\alpha_s \in \Z_{\geq 0},\,\theta_1,\ldots,\theta_s \in \K^\times
	\end{equation}
	form a basis of the $\K$-vector space of multi-linear recurrence sequences \cite[Theorem~2.1,\,Lemma~2.2]{Schmidt}. In that case, the set $\{\bxi_\omeg,\, \omeg \in \Omega\}$ coincides with the $\K$-vector space spanned by the Hahn series introduced in \cite{FR25}. It is not difficult to adapt Algorithm \ref{algo:matrixH} so that it returns the entries of $H$ as linear combinations of some $\bxi_{(\u,\a)}$ with $\u$ of the form \eqref{eq:form_lin_rec2}. Indeed, using Lemma \ref{lem:Hahn}, the only delicate task is to express $\u^{[\theta]}$ under this form when the general term of $\u$ satisfies \eqref{eq:form_lin_rec2}. Under this assumption, the general term of $\u^{[\theta]}$ is
	$$
u^{[\theta]}_{l,k_2,\ldots,k_s}=
\left(\sum_{k_1=1}^{l-1} k_1^{\alpha_1}(\theta_1\theta^{-1})^{k_1} \right)
k_2^{\alpha_2}\cdots k_s^{\alpha_s}\theta^l\theta_2^{k_2}\cdots \theta_s^{k_s}\,.
	$$
	Thus, one only has to express $\sum_{k=1}^{l-1} k^{\alpha_1}(\theta_1\theta^{-1})^{k}$ under the desired form. When $\theta_1\theta^{-1}=1$, Faulhaber's formula describes it as an explicit polynomial of degree $\alpha_1+1$ in $l$. When $\theta_1\theta^{-1}\neq 1$, one can find a polynomial $P(X)$ and a number $\gamma \in \K$, depending on $\alpha_1,\theta_1,\theta$, such that $\sum_{k=1}^{l-1} k^{\alpha_1}(\theta_1\theta^{-1})^{k}=P(l)(\theta_1\theta^{-1})^l+\gamma$. Explicit formulas may easily be found by applying $\alpha_1$ times the operator $\lambda\frac{\partial}{\partial \lambda}$ on both sides of the identity $	\sum_{k=1}^{l-1} \lambda^k = \frac{\lambda^{l}-1}{\lambda - 1}$	and setting $\lambda:= \theta_1\theta^{-1}$.	
	\end{rem*}
	
To conclude this section, and for the sake of completeness, let us explain why the Hahn series $\bxi_\omeg$ are solutions of $p$-Mahler equations. Since this result is not necessary to prove Theorem~\ref{thm:algo_solutions_equations} and Theorem~\ref{thm:algo_solutions_systems}, we shall be brief. When the characteristic of $\K$ is $0$, this result is proved in \cite{FR25} using the fact that any multi-linear recurrence sequence is a linear combination of sequences with general term $k_1^{\alpha_1}\cdots k_s^{\alpha_s}\lambda_1^{k_1}\cdots\lambda_s^{k_s}$. This approach does not work in positive characteristic (consider, for example, the sequence $\lfloor k/p\rfloor \mod p$ in $\mathbb F_p$). Fix a Hahn series $\bxi_\omeg$. Let $s$ be the associated parameter. If $s=0$ then $\bxi_{\omeg}=1$ is a solution of a $p$-Mahler equation. Suppose that $s\geq 1$ and that the result holds for any $\bxi_{\omeg'}$ with $s'<s$. 
    Keeping the notation of Lemma~\ref{lem:lin_rec}, it is straightforward that
   $$
    b_{1,0}\bxi_\omeg(z)+b_{1,1}\bxi_{\omeg}(z^p)+\cdots + b_{1,m_1}\bxi_{\omeg}\left(z^{p^{m_1}}\right)
   $$
   is a linear combination of some $z^{-\gamma} \bxi_{\omeg'}$ where the parameter $s'$ associated to each $\omeg'$ is $s-1$. By induction hypothesis, $\bxi_{\omeg}$ is a solution of an inhomogeneous equation whose right-hand side is itself a solution of a $p$-Mahler equation. It is classical that $\bxi_\omeg$ is then a solution of a $p$-Mahler equation.

	\subsection{Computing solutions of constant systems}\label{sec:constantMatrix}
	 Let $C$ denote a non-singular constant matrix. The last step consists in computing a matrix $e_C$ whose entries are $\K$-linear combinations of the $e_c\logm^j$ such that $\puip(e_C)=Ce_C$. The strategy is explicitly described in \cite{FR25}. Let us recall it briefly. Let $C=DU$ be the multiplicative Dunford-Jordan decomposition of $C$, where $D$ is diagonalizable, $U$ is unipotent, and $D$ and $U$ commute. Set $\logm^{[k]} := \frac{\logm(\logm-1)\cdots(\logm-k+1)}{k!}$ when $k \in \Z_{\geq 0}$. Then, since $(U-{\rm I})^m=0$, the matrix
	$$
	e_U := \sum_{k=0}^{m-1} \logm^{[k]} (U-{\rm I})^k
	$$
	satisfies $\puip(e_U)=Ue_U$. Let $Q \in \GL_m(\K)$ be such that $QDQ^{-1} = {\rm diag}(c_1,\ldots,c_m)$, with $c_1,\ldots,c_m \in \K^\times$. Then, $e_D:=Q{\rm diag}(e_{c_1},\ldots,e_{c_m})Q^{-1}$ satisfies $\puip(e_D)=De_D$. Eventually, since $D$ and $U$ commute, $e_U$ and $e_D$ commute and $e_C:=e_De_U$ satisfies $\puip(e_C)=Ce_C$. This leads to the following algorithm.
	
	\medskip
	\begin{algorithm}[H]
		\label{algo:constante}
		\SetAlgoLined
		\KwIn{A constant matrix $C$ with coefficients in $\K$.}
		\KwOut{The matrix $e_C$.}
		Write $C=DU$\;
		Set $e_U := \sum_{k=0}^{m-1} \logm^{[k]} (U-{\rm I})^k$\;
		Let $Q \in \GL_m(\K)$ be such that $QDQ^{-1} = {\rm diag}(c_1,\ldots,c_m)$ is diagonal\;
		Set $e_D := Q{\rm diag}(e_{c_1},\ldots,e_{c_m})Q^{-1}$\;
		\Return $e_De_U$\;
		\caption{An algorithm to compute a fundamental matrix of solution to a constant system}
	\end{algorithm}
	
	\subsection{On the necessity for the base field to be algebraically closed}\label{sec:alg_clo}
	We claim that, when $\K$ is not algebraically closed, it is not always possible to build a matrix $e_C$ whose entries are linear combinations of some $e_c\logm$ with $c\in \K^\times$. Let us justify this. Over the base field $\mathbb Q$ we consider the matrix
    $$
C:=\begin{pmatrix}
 	0 & 1 \\ -1 & 0 
 	\end{pmatrix}
 	$$
    Algorithm \ref{algo:constante} returns the matrix
    $$
    e_C:=
\begin{pmatrix}
\frac{1}{2}(e_i+e_{-i}) & 
 	\frac{i}{2}(e_{-i}-e_i)  \\ \frac{i}{2}(e_{i}-e_{-i}) &  \frac{1}{2}(e_i+e_{-i})
 	\end{pmatrix}
    $$
    Since the elements $e_c$, $c \in \Q$, are linearly independent over $\Hahn_{\Q}[\logm]$ and since the solution $e_C$ of the system $\puip(Y)=CY$ is unique up to right product by an element of $\GL_m(\Q)$, there exists no such solution whose entries are linear combinations of the elements $e_c\logm^j$, $c \in \mathbb Q$, $j \in \Z_{\geq 0}$.

	\subsection{Description of the main algorithm}
	
	Recall that, to any equation \eqref{eq:Mahler}, we can associate a Mahler system by considering the companion matrix
	$$
	A := \begin{pmatrix}
	& 1 
	\\ && \ddots 
	\\& & & 1
	\\ -\frac{a_0}{a_m} & \cdots & \cdots & - \frac{a_{m-1}}{a_m}\,.
	\end{pmatrix}.
	$$
	Then, finding a basis of solutions to \eqref{eq:Mahler} amounts to computing a fundamental matrix of solutions to the associated system. This is the strategy which justifies Algorithm \ref{algo:basis_solutions}.
	
	\medskip
	
	\begin{algorithm}[H]\label{algo:basis_solutions}
		\SetAlgoLined
		\KwIn{A Mahler equation \eqref{eq:Mahler} and an integer $n$}
		\KwOut{A basis of solution in the form \eqref{eq:basis_solutions}, with a truncation of the series $f_{i,c,j,\omeg}$ up to order $n$}
		Let $A$ denote the companion matrix associated to this equation\;
		Compute $d:=d(A)$\;
		\If{$d\neq 1$}{replace $A$ with $A(z^{d})$\;}
		Let $(\overline{P},\Theta)$ be the output of Algorithm \ref{algo:solutions_blocks} when taking $A$ as input\;
		Compute the coefficients $P_{k}$ of the Puiseux matrix $P$ up to order $dn$ thanks to \eqref{eq:formule_blocks}\;
		Replace $\overline{P}$ with this new matrix\;
		Compute $Q \in \GL_m(\KK)$ such that $Q\Theta Q^{-1}$ is upper triangular\;
		Replace $\overline{P}$ with $\overline{P}Q^{-1}$ and $\Theta$ with $Q\Theta Q^{-1}$\;
		Let $H$ be the output of Algorithm \ref{algo:matrixH} taking $\Theta$ as input\;
		Let $e_C$ be the output of Algorithm \ref{algo:constante} taking the constant term of $\Theta$ as input\;
		\Return the first row of $\overline{P}(z^{1/d})H(z^{1/d})e_C$.
		\caption{An algorithm to compute a basis of solution to a Mahler equation}
	\end{algorithm}
	\begin{rem*}
		In the last step of this algorithm, one has to deal with terms of the form $\bxi_{\omeg}(z^{1/d})$. However, up to replace the tuple $\a$ with $\a/d$, such a term is of the form $\bxi_{\omeg'}$ for some $\omeg'\in \Omega$. Thus, Algorithm \ref{algo:basis_solutions} returns a basis of solutions of the desired form. 
		
		As already mentioned, the decomposition \eqref{eq:basis_solutions} is not unique. Nevertheless, when $\K$ has characteristic $0$, it is possible to make it unique. Consider the set $\Omega_0$ of pairs $\omeg:=(\u,\a)$ where
		\begin{itemize}
			\item $\u$ has a general term of the form $k_1^{\alpha_1}\cdots k_s^{\alpha_s}\lambda_1^{k_1}\cdots \lambda_s^{k_s}$, where $\alpha_1,\ldots,\alpha_s \in \Z_{\geq 0}$ and $\lambda_1,\ldots,\lambda_s \in \KK^\times$;
			\item  $a$ is a tuple of positive rational numbers whose denominator is coprime with $p$ and whose numerator is not divisible by $p$ (recall that $p$ is not necessarily prime).
		\end{itemize}
		A decomposition \eqref{eq:basis_solutions} for which each $\omeg$ belongs to $\Omega_0$ is called \textit{standard} in \cite{FR25} and such a decomposition always exists. Passing from any decomposition to the standard one is not difficult. Details are provided in \cite[Section 5.1.2]{FR25}. 
	\end{rem*}

	\subsection{Computing an equation for each Puiseux series appearing in \eqref{eq:basis_solutions}}
	
	Algorithm \ref{algo:basis_solutions} returns a list $(y_1,\ldots,y_m)$ of solutions of a $p$-Mahler equation \eqref{eq:Mahler_eq}, of the form
	$$
	y_i= \sum_{c \in K_0} \sum_{j =0}^{j_0} \sum_{\omeg \in \Omega_1} f_{i,c,j,\omeg}(z)\bxi_{\omeg} e_c \logm^j
	$$
	where $K_0,\Omega_1$ are finite sets and $f_{i,c,j,\omeg}$ are Puiseux series whose expansion up to $z^n$ is computed. As mentioned in the introduction, each series can be uniquely determined by providing a $p$-Mahler equation it satisfies. The construction is as follows. These Puiseux series $f_{i,c,j,\omeg}$ are explicit $\KK$-linear combinations of the entries of the matrix $P$, the first matrix of an admissible pair $(P,\Theta)$. The equation satisfied by a sum of $p$-Mahler Puiseux series can be easily determined from the equations satisfied by each term. Hence, one only has to compute an equation for each entry of $P$. Let $f_{i,j}$ be the $(i,j)$-th entry of $P$. From $\puip(P)\Theta=AP$ we deduce that for each $k$, 
	$$
	f_{i,j}(z^{p^k})=\left(\puip^{k-1}(A)\cdots \puip(A)AP \Theta^{-1}\puip(\Theta)^{-1}\cdots \puip^{k-1}(\Theta)^{-1}\right)_{i,j}
	$$
	Thus, the $m^2+1$ series $f_{i,j}(z),\ldots,f_{i,j}\left(z^{p^{m^2}}\right)$ are explicit linear combinations of the $m^2$ entries of $P$. One readily deduces a $p$-Mahler equation satisfied by $f_{i,j}$.

	\section{Examples}\label{sec:RudinShapiro}
	
	Let us first run Algorithm \ref{algo:basis_solutions} on an example in characteristic $0$. Consider the Rudin-Shapiro equation
	\begin{equation}\label{eq:RS}
	y(z)+(z-1)y(z^2)-2zy(z^4)=0\,.
	\end{equation}
	Here we have $p=2$ and we may take $\K=\Q$. We run Algorithm \ref{algo:basis_solutions} with this equation as input and the integer $n:=9$. It computes the companion matrix
	$$
	A(z)=\begin{pmatrix}0& 1 \\ \frac{1}{2z} & \frac{z-1}{2z}\end{pmatrix}.
	$$
	associated with \eqref{eq:RS}. The associated Newton polygon is the convex hull of the set $\{(1,i),(2,j),(4,k)\,:\,i,j\geq 0,\,k\geq 1\}$. Its slopes are $0$ and $\frac{1}{2}$ so that $d(A)=1$.
	
	Our algorithm calls Algorithm \ref{algo:solutions_blocks} which returns a pair $(\overline{P},\Theta)$ with
	$$
	\overline{P}=\begin{pmatrix}
	1+z & z \\ 1 & \frac{1}{z}-1+z
	\end{pmatrix},\, \text{ and } \,
	\Theta=\begin{pmatrix}
	1 & \frac{1}{z}-1 \\ 0 & -\frac{1}{2}
	\end{pmatrix}\,.
	$$
	as we will prove in Section \ref{sec:call_algoblocks}. Letting $\overline{P_1},\overline{P_2}$ denote the columns of $\overline{P}$, our algorithm then computes recursively the coefficients $P_{1,n}$ and $P_{2,n}$ for $n\in \{2,\ldots,9\}$ thanks to \eqref{eq:formule_blocks}. Then it replaces the matrix $\overline{P}$ with the matrix
	$$
	\footnotesize{\setlength{\arraycolsep}{3pt} \medmuskip=1mu
		\begin{pmatrix}
		1+z+z^2-z^3+z^4+z^5-z^6+z^7+z^8+z^9 &z - \frac{5}{2}z^2 + \frac{3}{2}z^3 + \frac{5}{4}z^4 - \frac{7}{4}z^5 + \frac{5}{4}z^6 - \frac{1}{4}z^7 - \frac{5}{8}z^8 + \frac{3}{8}z^9 \\ 1+z^2+z^4-z^6+z^7+z^9 &\frac{1}{z} - 1 + z - \frac{3}{2}z^2 + z^3 + \frac{1}{4}z^4 - z^5 + \frac{1}{4}z^6 + z^7 - \frac{13}{8}z^8 + z^9
		\end{pmatrix}}
	$$
	It then calls Algorithms \ref{algo:matrixH} and \ref{algo:constante}, which return (see Sections~\ref{sec:call_algoH} and \ref{sec:call_algoC}) 
	$$H = \begin{pmatrix}
	1 & \bxi_{([-2],1)} \\ 0 & 1
	\end{pmatrix}\, \text{ and } \, e_C=\begin{pmatrix} 1&\frac{2}{3}e_{\frac{-1}{2}}-\frac{2}{3}\\ 0&e_{-\frac{1}{2}}\end{pmatrix}\,,
	$$
	where $[-2]=((-2)^k)_k \in \Q^\Z$. Then, Algorithm \ref{algo:basis_solutions} returns the entries of the first row of $\overline{P}He_C$, that is
	$$
	f_1:=1+z+z^2-z^3+z^4+z^5-z^6+z^7+z^8+z^9
	$$
	and
	$f_2:=-\frac{2}{3}f_1 +\left(f_1\bxi_{([-2],1)}+\frac{2}{3}f_1 +g\right) e_{-\frac{1}{2}}$
	where $g$ is the top-right entry of $\overline{P}$.

	\subsection{Calling Algorithm \ref{algo:solutions_blocks}}\label{sec:call_algoblocks}
	Since $\val A=-1$ and $\val A^{-1}=0$, the parameters~\eqref{eq:nota_val} are
	$\nu_P = -1$, $\nu_\Theta  = -1$, $\nu=-3$, and $\mu=1$. Write $A^{-1}(z) = B_0 + B_1z$ with 
	$$
	B_0:=\begin{pmatrix}
	1 & 0 \\ 1 & 0
	\end{pmatrix} \quad \text{and} \quad B_1:=\begin{pmatrix}
	-1 & 2 \\ 0 & 0 
	\end{pmatrix}\, .
	$$
	One checks that $M=M_0$ and $M_1$ have the following block-decomposition
	$$
	M_0 =\tiny\begin{pmatrix}
	0 & 0 & 0 & 0 & 0 \\
	0 & 0 & B_0 & 0 & 0 \\
	0 & 0 & B_1 & 0 & 0 \\
	0 & 0&0 & B_0 & 0 \\
	0 & 0 &0 & B_1 & 0 
	\end{pmatrix} \quad \normalsize\text{and} \quad 
	\ \normalsize M_{-1}=\tiny\begin{pmatrix}
	0 & 0& B_0 & 0 & 0\\
	0 & 0 & B_1 & 0& 0 \\
	0 & 0 & 0 & B_0 & 0 \\
	0 & 0 & 0 & B_1 & 0 \\
	0 & 0 & 0 & 0 & B_0 \\
	\end{pmatrix}
	$$
	where $0$ denotes the zero matrix of size $2\times 2$.  
	Algorithm \ref{algo:solutions_blocks} is initialized with $\X:=\{0\}$, $\mathcal E:=()$ and $\mathcal U:=()$. \subsubsection*{First iteration of the main `while' loop} Entering the main `while' loop, it sets $\mathfrak U:=\{0\}$, $\mathfrak F:=V_0$ and computes
	\begin{align*}
	\mathfrak G&:=V_0 \cap M^{-1} V_0\cap MV_0 =\{\prescript{\rm t}{}{(0,0,0,0,0,0,\lambda_1,\lambda_1,\lambda_2,0)}\,:\,\lambda_1,\lambda_2\in\QQ\}
	\end{align*}
	\indent $\blacktriangleright$\textit{Execution of the inner `while' loop} 
	\begin{itemize}
		\item Since $\mathfrak G\neq \mathfrak F$, it enters the inner `while' loop and it sets $\mathfrak F:= \mathfrak G$. Then, it computes the new space $\mathfrak G$:
		$$
		\mathfrak G := \mathfrak F \cap M^{-1} \mathfrak F\cap M\mathfrak F = {\rm span}_\QQ\{\e_1\}. 
		$$
		where $\e_1:=\prescript{\rm t}{}{(0,0,0,0,0,0,1,1,1,0)}$.
		\item Since $\mathfrak G\neq \mathfrak F$, it enters this inner loop a second time. The computation at that stage leaves $\mathfrak G$ unchanged. Thus, our algorithm exits this loop.
	\end{itemize}
	
	Since $\X=\mathfrak U=\{0\}$, it sets $\mathfrak Y:=\{0\}$ and $\mathfrak Z:={\rm span}_\QQ\{\e_1\}$. Thus $E_{\mathfrak Y}$ is the empty matrix and $E=E_{\mathfrak Z}=\e_1$. The algorithm appends $\e_1$ to $\mathcal E$ and the zero space to $\mathcal U$ so that $\mathcal E=(\e_1)$ and $\mathcal U=(\{0\})$. Then it sets $\X={\rm span}_\QQ\{\e_1\}$. Since $\dim(\X)=1< 2$, it enters a second time the main `while' loop. 
	
	\subsubsection*{Second iteration of the main `while' loop}
	It sets 
	$\mathfrak U:={\rm span}_{\QQ}(M_0\e_1,M_1\e_1)$ which is the $\QQ$-vector space spanned by $\e_1$ and $\x$
	where
	$$
	\x:=\prescript{\rm t}{}{(0,0,0,0,1,1,1,0,1,1)}\,.
	$$
	Then it sets $\mathfrak F:=V_0$. Then, the first computation of $\mathfrak G$ gives
	$$
	\mathfrak G:=\{\prescript{\rm t}{}{(0,0,0,0,\lambda_0,\lambda_1,\lambda_1+\lambda_2,\lambda_2,\lambda_3,\lambda_1)}\,:\,\lambda_0,\lambda_1,\lambda_2,\lambda_3\in \QQ\}
	$$
	\indent $\blacktriangleright$\textit{Execution of the inner `while' loop} 
	\begin{itemize}
		\item
		Since $\mathfrak G \neq \mathfrak F$, it enters the inner `while' loop and it sets $\mathfrak F:= \mathfrak G$. Then, $\mathfrak G$ becomes
		$$
		\mathfrak G := \{\prescript{\rm t}{}{(0,0,0,0,0,\lambda_1,\lambda_1+\lambda_2,\lambda_2,\lambda_3,\lambda_1)}\,:\,\lambda_1,\lambda_2,\lambda_3\in \QQ\}\,.
		$$
		\item One more iteration of this `while' loop gives $\mathfrak G:=
		{\rm span}_\QQ\{\e_1,\e_2\}$ with
		$$
		\e_2:=\prescript{\rm t}{}{(0,0,0,0,0,1,0,-1,1,1)}.
		$$
		\item Then, one last iteration leaves $\mathfrak G$ unchanged, so that it exits this loop.
	\end{itemize}
	
	Since $\e_1,\e_2$ and $\x$ are linearly independent, we have $\mathfrak Y=\{(0)\}$ and $\mathfrak Z={\rm span}_\QQ\{\e_2\}$. Thus $E_{\mathfrak Y}$ is the empty matrix and $E=E_{\mathfrak Z}=\e_2$. Then, the algorithm appends $\e_2$ to $\mathcal E$ and the space $\mathfrak U$ to $\mathcal U$ so that $\mathcal E=(\e_1,\e_2)$ and $\mathcal U=(\{0\}, {\rm span}_\QQ\{\e_1,\x\})$. Then it sets $\X:={\rm span}_\QQ\{\e_1,\e_2\}$. Since $\dim X=2$, it leaves the main `while' loop.
	
	\subsubsection*{Execution of the `for' loop}
	We have $r=2$. The algorithm enters the `for' loop. All blocks of $\Theta$ will be $1\times 1$ matrices. By some abuse of notation we will write them as elements of $\QQ$. When $j=1$, the algorithm sets $\Theta_1:=1$ and
	$$
	\overline{P_1}=\begin{pmatrix}
	1+z \\ 1 
	\end{pmatrix}\,.
	$$
	In the second iteration of the loop (with $j=2$), it sets $\Theta_2 := -1/2$, $\Theta_{1,2,0}:=-1$, $\Theta_{1,2,-1}:=1$ and
	$$
	\overline{P_2}=\begin{pmatrix}
	z \\ \frac{1}{z}-1+z
	\end{pmatrix}\,.
	$$
	It exits the `for' loop and returns $((\overline{P_1},\overline{P_2}),\Theta)$ where
	$$\Theta:=\begin{pmatrix}
	1 & \frac{1}{z}-1 \\ 0 & -\frac{1}{2}
	\end{pmatrix}\,.
	$$

	\subsection{Calling Algorithm \ref{algo:matrixH}}\label{sec:call_algoH}
	Algorithm \ref{algo:matrixH} takes as input the matrix $\Theta$. The only entry of $H$ this algorithm has to compute is $h_{1,2}$ which satisfies
	$$
	-\frac{1}{2}h_{1,2}(z^2) - h_{1,2}(z)=\frac{1}{z}\,.
	$$
	This is an equation of the form \eqref{eq:inhomogene} with $\kappa:=-\frac{1}{2}$ and $\eta:=\tau:=1$. Thanks to Lemma \ref{lem:Hahn}, it returns $h_{1,2}(z):=\bxi_{([-2],1)}=\sum_{k\geq 1}(-2)^kz^{-1/2^k}$ and 
	$$
	H:= \begin{pmatrix}
	1 & \bxi_{([-2],1)} \\ 0 & 1
	\end{pmatrix}.
	$$
	\subsection{Calling Algorithm \ref{algo:constante}}\label{sec:call_algoC}
	Algorithm \ref{algo:constante} is called with $C=:\begin{pmatrix}
	1 & -1 \\ 0 & -\frac{1}{2}
	\end{pmatrix}$ as an input. The $DU$ decomposition of $C$ is $U:={\rm Id}$ and $D:=C$. After diagonalizing $C$, Algorithm \ref{algo:constante} returns $e_C$ where
	$$
	e_C=
	\begin{pmatrix} 1&1 \\ 0&-\frac{3}{2}\end{pmatrix}  \begin{pmatrix}
	1 & 0 \\ 0 & e_{-\frac{1}{2}}  
	\end{pmatrix}\begin{pmatrix} 1&\frac{2}{3} \\ 0&-\frac{2}{3}\end{pmatrix}=\begin{pmatrix} 1&\frac{2}{3}e_{\frac{-1}{2}}-\frac{2}{3}\\ 0&e_{-\frac{1}{2}}\end{pmatrix}\,.
	$$

	\subsection{An equation for the second Puiseux term}
Algorithm \ref{algo:basis_solutions} returns $f_1$ and $f_2$, which are truncations at order $n=9$ of two solutions $h_1$ and $h_2$ to \eqref{eq:RS}. The solution $h_1$ is a power series solution to \eqref{eq:RS}. By contrast, the solution $h_2$ is of the form 
	$$
	h_2 = -\frac{2}{3}h_1 +\left(h_1\bxi_{([-2],1)}+\frac{2}{3}h_1 +h\right) e_{-\frac{1}{2}}
	$$
	where $h$ is some power series. Precisely, if $P$ is the matrix of Laurent series whose truncation is $\overline{P}$ and who satisfies $\puip(P)\Theta=AP$, then, $h$ is the top-right entry of $P$. Let us explain how to compute a Mahler equation satisfied by $h$. We let $P_{i,j}$ denote the entries of $P$. In particular, $h_1=P_{1,1}$ and $h=P_{1,2}$. Iterating the identity $\puip(P)\Theta=AP$ and isolating the terms $h(z^{2^k})$, $k\in\{0,\ldots,4\}$, we obtain
	$$
	\scriptsize
	\begin{pmatrix}
	h(z)\\h(z^2)\\h(z^4)\\h(z^8)\\h(z^{16})
	\end{pmatrix}\!\!
	=\!\!
	 \fontsize{4}{5}\selectfont
	 \setlength{\arraycolsep}{0.4pt}
	 \left(\begin{array}{c|c|c|c}
	 0 & 1 & 0 & 0 \\[0.8em] \hline 
	 0 & 0 & \frac{2(1-z)}{z} & -2 \\[0.8em] \hline
	 \frac{(z-1)^2}{z^3} & \frac{2}{z} & \frac{(z-1)^3}{z^3} & \frac{2(z-1)}{z} \\[0.8em] \hline
	 \begin{array}{c}
	 \frac{-3z^6+4z^5+z^4}{2z^7} \\
	 \frac{-4z^3+3z^2-1}{2z^7}
	 \end{array} & \frac{2(1-z^2)}{z^3} & \begin{array}{c}
	 \frac{-3z^7+7z^6-9z^5+3z^4}{2z^7} \\
	 \frac{+3z^3-3z^2+z+1}{2z^7}
	 \end{array} & \frac{-2(z^3-z^2+z+1)}{z^3} \\[1.5em] \hline
	 \begin{array}{c}
	 \frac{5z^{14}-8z^{13}-z^{12}+8z^{11}}{4z^{15}} \\
	 \frac{-z^{10}-8z^9+11z^8-8z^7}{4z^{15}} \\
	 \frac{+3z^6-3z^4+z^2+1}{4z^{15}}
	 \end{array} & \frac{2(z^6-z^4+z^2+1)}{z^7} & \begin{array}{c}
	 \frac{5z^{15}-13z^{14}+17z^{13}-7z^{12}}{4z^{15}} \\
	 \frac{-z^{11}-7z^{10}+5z^9-3z^8}{4z^{15}} \\
	 \frac{+3z^7-3z^6+3z^5+3z^4}{4z^{15}} \\
	 \frac{+z^3-z^2-z-1}{4z^{15}}
	 \end{array} & \begin{array}{c}
	 \frac{2z^7-2z^6+2z^5}{z^7} \\
	 \frac{+2z^4+2z^3}{z^7} \\
	 \frac{-2z^2-2z-2}{z^7}
	 \end{array}
	 \end{array}\right)
\scriptsize\!\!
	\begin{pmatrix}
	P_{1,1}(z)\\ P_{1,2}(z) \\P_{2,1}(z)\\ P_{2,2}(z)
	\end{pmatrix}
	$$

	Looking at a vector in the left-kernel of this matrix, we get
	
	 {\medmuskip=1mu
	\thickmuskip=2mu
	\thinmuskip=1mu
	\footnotesize
	\begin{align*}
(-2z^{10}-6z^{9}-8z^{8}-8z^{7}-4z^{6}+4z^{5}+8z^{4}+8z^{3}+6z^{2}+2z)&h(z)
		\\ + 
	(-z^{11}-z^{9}-2z^{8}+2z^{7}+2z^{5}+4z^{4}-z^{3}-z-2)&h(z^2)
	\\ +(z^{13}+z^{12}+z^{11}+3z^{10}+5z^{9}+5z^{8}-2z^{7}-6z^{6}-z^{5}-z^{4}-3z^{3}-z^{2}-z-1)&h(z^4)
	\\+( 	-z^{13}-z^{12}+3z^{11}+3z^{10}-4z^{7}-4z^{6}+z^{5}+z^{4}+z^{3}+z^{2})&h(z^8)
	\\ +(-2z^{13}-2z^{12}+2z^{11}+2z^{10})&h(z^{16})=0,
	\end{align*}}
	which is a $2$-Mahler equation satisfied by $h$.

    \subsection{Carlitz zeta-function} Let us briefly present an example in positive characteristic.
    Let $p\geq 0$ be a prime number, $\theta$ be a transcendental element over $\mathbb F_p$ and $\K$ denote the complement of the algebraic closure of $\mathbb F_p((\theta^{-1}))$ for the valuation associated to $\theta^{-1}$. As mentioned in Section~\ref{sec:motivations}, the values at $1$ of the Carlitz zeta-function, $\zeta_C$ is equal to $f(\theta)$ where $f(z) \in \K[[z]]$ satisfies
	$$
	f(z^p) + (z^p - \theta) f(z) = - (z^p - \theta)\,.
	$$
    Thus, $f(z)$ satisfies the following Mahler equation, of order $2$:
    \begin{equation}\label{eq:Carlitz}
    (z^p-\theta)(z^{p^2}-\theta)f(z) - (z^p - \theta - 1)(z^{p^2}-\theta)f(z^p) - (z^p-\theta)f(z^{p^2})=0\,.
    \end{equation}
    One easily checks that the Newton polygon associated to this equation has only one slope which is null. Thus $\nu_\Theta=\nu_P=\nu=\mu=0$. The matrix $M$ is then given by the constant term of the inverse of the matrix of this system, that is
    $$
    M=\begin{pmatrix}
        1 + \theta^{-1} & - \theta^{-1} \\ 1 & 0
    \end{pmatrix}
    $$
    Thus, it acts as an isomorphism on $V_0:=\K^2$. Then, Algorithm \ref{algo:solutions_blocks} returns
    $$
    \overline{P}=\begin{pmatrix}
        1 & 1 \\ 1 & \theta
    \end{pmatrix}  \quad \text{and} \quad \Theta=\begin{pmatrix}
        1 & 0 \\ 0 & \theta
    \end{pmatrix}.
    $$
    Algorithm \ref{algo:matrixH} returns $H={\rm I}_2$ and Algorithm \ref{algo:constante} returns
    $$
   e_C=\begin{pmatrix}
        1 & 0 \\ 0 & e_{\theta}
    \end{pmatrix}
    $$
    Thus, a basis of solutions of \eqref{eq:Carlitz} is given by $f(z)$ and $g(z)e_\theta$ where
     where $f(z)$ is defined as above and $g(z)$ satisfies
     $$
    (z^p-\theta)(z^{p^2}-\theta)g(z) - \theta (z^p - \theta - 1)(z^{p^2}-\theta)g(z^p) - \theta^2 (z^p-\theta)g(z^{p^2})=0\,.
    $$

\end{document}